\documentclass[11pt,letterpaper]{article}
\usepackage[utf8]{inputenc}
\usepackage[margin=1in]{geometry}
\usepackage{palatino}
\usepackage{macros}
\usepackage{eulervm}

\usepackage{braket}
\usepackage{amsfonts}
\usepackage{booktabs}
\usepackage{multirow}
\usepackage{diagbox}
\usepackage{algorithm}
\usepackage{algpseudocode}

\usepackage[dvipsnames]{xcolor}
\hypersetup{
	colorlinks=true,
	pdfpagemode=UseNone,
    citecolor=OliveGreen,
    linkcolor=NavyBlue,
    urlcolor=Magenta,
	pdfstartview=FitW
}

\title{Quantum Complexity of Weighted Diameter and Radius in CONGEST Networks}
\author{
    Xudong Wu\thanks{Nanjing University. Email: xdwu@smail.nju.edu.cn}
    \and Penghui Yao\thanks{Nanjing University. Email: pyao@nju.edu.cn}
}
\date{\today}

\begin{document}

\maketitle

\begin{abstract}
This paper studies the round complexity of computing the weighted diameter and radius of a graph in the quantum CONGEST model.
We present a quantum algorithm that $(1+o(1))$-approximates the diameter and radius with round complexity $\widetilde O\left(\min\left\{n^{9/10}D^{3/10},n\right\}\right)$, where $D$ denotes the unweighted diameter.
This exhibits the advantages of quantum communication over classical communication since computing a $(3/2-\varepsilon)$-approximation of the diameter and radius in a classical CONGEST network takes $\widetilde\Omega(n)$ rounds, even if $D$ is constant \cite{AbboudCK16}.
We also prove a lower bound of $\widetilde\Omega(n^{2/3})$ for $(3/2-\varepsilon)$-approximating the weighted diameter/radius in quantum CONGEST networks, even if $D=\Theta(\log n)$.
Thus, in quantum CONGEST networks, computing weighted diameter and weighted radius of graphs with small $D$ is strictly harder than unweighted ones due to Le Gall and Magniez's $\widetilde O\left(\sqrt{nD}\right)$-round algorithm for unweighted diameter/radius~\cite{GallM18}.
\end{abstract}

\section{Introduction}

Quantum distributed computing has received great attention in the past decade~\cite{Ben-OrH05,TaniKM12,ElkinKNP14,GallM18,GallNR19,IzumiG19,IzumiGM20,Censor-HillelFG22,MagniezN20,GavoilleKM09}.
A large body of work has been devoted to investigating the quantum advantages in distributed computing.
In this paper, we are concerned with the CONGEST networks, which are one of the most fundamental models in distributed computing.
In a classical CONGEST network, the nodes synchronously exchange classical messages, and each channel has $O(\log n)$-bit bandwidth, where $n$ is the number of nodes in the network.
Quantum CONGEST networks were first introduced by  Elkin, Klauck, Nanongkai,
and Pandurangan~\cite{ElkinKNP14}, where the only difference is that the nodes exchange quantum messages and the bandwidth of each channel is $O(\log n)$ qubits. The round complexity of diameter and radius  of unweighted graphs in classical CONGEST networks has been extensively studied~\cite{AbboudCK16,AnconaCDEW20,HolzerW12,PelegRT12,FrischknechtHW12,HolzerPRW14}.
Le Gall and Magniez~\cite{GallM18} proved that quantum communication may save the round complexity in CONGEST networks if the graph has a low diameter.

In this paper, we further investigate the round complexity of computing the diameters and radius of weighted graphs in quantum CONGEST networks. We prove that quantum communication may also save the round complexity for both problems.

\subsection{Our Results}
\label{sec:results}

The following is one of our main results which asserts that quantum communication may save the round complexity for computing the weighted diameter and radius of a graph given that the graph has a low unweighted diameter.

\begin{theorem}
There exists a $\widetilde O\left(\min\left\{n^{9/10}D^{3/10},n\right\}\right)$-round distributed algorithm computing a $(1+o(1))$-approximation of the weighted diameter/radius with probability at least $1-1/\text{poly}(n)$, in the quantum CONGEST model, where $D$ denotes the unweighted diameter.
\label{thm:upper_bound}
\end{theorem}

Holzer and Pinsker in~\cite{HolzerP15} and Abboud, Censor-Hillel and Khoury in~\cite{AbboudCK16} proved that $(2-o(1))$-approximating the diameter and $(3/2-\varepsilon)$-approximating the (even unweighted) radius in the classical CONGEST network require $\widetilde\Omega(n)$ rounds, even when $D$ is constant.
Therefore, Theorem~\ref{thm:upper_bound} exhibits the advantages of quantum communication over classical communication in approximating the weighted diameter/radius when $D=o(n^{1/3})$.

We prove Theorem~\ref{thm:upper_bound} by applying the framework of distributed quantum optimization introduced by Le Gall and Magniez in~\cite{GallM18}.
Note that the diameter and radius are the maximum and minimum of eccentricities respectively.
It will not give a sublinear-time algorithm if we simply apply a quantum search algorithm, because evaluating the eccentricity of one node takes $\widetilde\Theta\left(\sqrt n\right)$ rounds (the lower bound is due to \cite{ElkinKNP14}), and the searching process should require another $\widetilde\Theta\left(\sqrt n\right)$ times of evaluation (the number of nodes with maximum/minimum eccentricity maybe $O(1)$).
Thus the number of rounds in total will be $\widetilde{\Theta}(n)$.\footnote{
$\widetilde O(\cdot)$, $\widetilde\Omega(\cdot)$, and $\widetilde\Theta(\cdot)$ hide polylogarithmic factors.}
Our algorithm is inspired by Nanongkai's algorithm~\cite{Nanongkai14STOC} for approximating the weighted shortest paths in a classical network.
The algorithm constructs several small vertex sets and searches the node achieving the maximum/minimum eccentricity within those sets, which turns out to be a good approximation of diameter/radius.
Our algorithm quantizes Nanongkai's algorithm using the standard technique \cite{Bennett89} and further combines with the framework of distributed quantum optimization in~\cite{GallM18}.

We also prove lower bounds for approximating weighted diameter and radius.

\begin{theorem}
Any algorithm computing a $(3/2-o(1))$-approximation of the weighted diameter/radius requires $\widetilde\Omega(n^{2/3})$ rounds, in the quantum CONGEST model, even when $D=\Theta(\log n)$.
\label{thm:lower_bound_raw}
\end{theorem}

The hardness of both problems is proved via the communication complexity of quantum Server models.
The Server model is a variant of two-party communication complexity models introduced in \cite{ElkinKNP14}.
Combining with the graph gadget in \cite{AbboudCK16}, we get a reduction from the communication complexity of certain read-once functions to the round complexity of approximating the weighted diameter and radius.
We further apply a lifting theorem of quantum communication complexity~\cite{ElkinKNP14} to obtain the desired lower bounds.

Compared with Le Gall and Magniez's algorithm~~\cite{GallM18} for unweighted diameter/radius with $\widetilde O\left(\sqrt{nD}\right)$ rounds,  Theorem~\ref{thm:lower_bound_raw} says that computing weighted diameter/radius is strictly harder than unweighted diameter/radius, when $D$ is small.
While in the classical setting, computing weighted and unweighted diameter/radius have the same round complexity $\Theta(n)$~\cite{AbboudCK16,BernsteinN19}\footnote{The lower bound is proved in~\cite{AbboudCK16}. The upper bound is followed from the $\widetilde O(n)$-round algorithm for the exact weighted All-Pairs Shortest Paths (APSP) in~\cite{BernsteinN19}.}.

\subsection{Related Works}

A series of works started with the distance computation in the classical CONGEST network.
Earlier, Frischknecht, Holzer, and Wattenhofer~\cite{FrischknechtHW12} showed that computing the diameter of an unweighted graph with constant diameter requires $\widetilde\Omega(n)$ rounds, which is tight up to logarithmic factors since even computing All-Pairs Shortest Paths (ASAP) on an unweighted graph can be resolved in $O(n)$ rounds \cite{HolzerW12,PelegRT12}.
Abboud, Censor-Hillelet, and Khoury~\cite{AbboudCK16} later gave the same lower bound of $\widetilde\Omega(n)$ for $(3/2-\varepsilon)$-approximating the diameter/radius in sparse networks.
Bernstein and Nanongkai~\cite{BernsteinN19} provided a $\widetilde O(n)$-round algorithm computing the exact APSP on any weighted graph.
As a result, computing unweighted diameter/radius and weighted diameter/radius (exactly or with a small approximation ratio) have an almost tight complexity of $\widetilde\Theta(n)$ in the classical CONGEST network.
If a larger approximation ratio is allowed, there are $\widetilde O(\sqrt n+D)$-round algorithms for $3/2$-approximating the diameter/radius on any unweighted graph \cite{HolzerPRW14,AnconaCDEW20}.
Besides, Chechik and Mukhtar~\cite{ChechikM20} showed a $\widetilde O(\sqrt nD^{1/4}+D)$-round algorithm computing Single-Source Shortest Paths (SSSP) exactly on any weighted graph, which also gives a $2$-approximation of the diameter/radius.

As for the quantum setting, while quantum computation offers advantages over classical computation in various settings such as query complexity and two-party communication complexity, the power of quantum computation in distributed computing has not been fully explored.
In the quantum CONGEST network, Elkin et al.~\cite{ElkinKNP14} gave negative results for several problems such as minimum spanning tree, minimum cut, and SSSP, i.e., quantum communication does not speed up distributed algorithms for these problems.
Le~Gall and Magniez~\cite{GallM18} presented a $\widetilde O\left(\sqrt{nD}\right)$-round algorithm computing the diameter/radius on any unweighted graph, along with a $\widetilde O\left(\sqrt[3]{nD}+D\right)$-round algorithm $3/2$-approximating the diameter.
They also proved a $\widetilde\Omega(\sqrt n+D)$-lower bound for computing the unweighted diameter, which was later improved to $\widetilde\Omega\left(\sqrt[3]{nD^2}+\sqrt n\right)$ by Magniez and Nayak~\cite{MagniezN20}.
The above results are listed on Table~\ref{tab:complexity}.

\begin{table}[t]
\centering
\scriptsize
\renewcommand\arraystretch{1.5}
\caption{Complexity of computing diameter and radius in the CONGEST model.}
\label{tab:complexity}
\begin{tabular}{cccccccc}
\toprule[1.5pt]
\multirow{2}{*}{Problem} & \multirow{2}{*}{Variant} & \multirow{2}{*}{Approx.} & \multicolumn{2}{c}{Upper bound $\widetilde O(\cdot)$} & ~ & \multicolumn{2}{c}{Lower bound $\widetilde\Omega(\cdot)$} \\
\cline{4-5} \cline{7-8}
~ & ~ & ~ & Classical & Quantum & ~ & Classical & Quantum \\
\midrule[1.5pt]
\multirow{7}{*}{diameter} & ~ & exact & $n$~\cite{HolzerW12,PelegRT12} & $\sqrt{nD}$~\cite{GallM18} & ~ & $n$~\cite{FrischknechtHW12} & $\sqrt[3]{nD^2}+\sqrt n$~\cite{MagniezN20} \\
~ & ~ & $3/2-\varepsilon$ & $n$ & $\sqrt{nD}$ & ~ & $n$~\cite{AbboudCK16} & $\sqrt n+D$~\cite{GallM18} \\
~ & ~ & $3/2$ & $\sqrt n+D$~\cite{HolzerPRW14,AnconaCDEW20} & $\sqrt[3]{nD}+D$~\cite{GallM18} & ~ & open & open \\
~ & weighted & exact & $n$~\cite{BernsteinN19} & $n$ & ~ & $n$ & $n^{2/3}$ \\
~ & weighted & $(1,3/2)$ & $n$ & \textcolor{red}{$\min\left\{n^{9/10}D^{3/10},n\right\}$ (This work)} & ~ & $n$ & \textcolor{red}{$n^{2/3}$ (This work)} \\
~ & weighted & $2-\varepsilon$ & $n$ & $\min\left\{n^{9/10}D^{3/10},n\right\}$ & ~ & $n$~\cite{HolzerP15} & $\sqrt n+D$ \\
~ & weighted & $2$ & $\sqrt nD^{1/4}+D$~\cite{ChechikM20} & $\sqrt nD^{1/4}+D$ & ~ & open & open \\
\midrule
\multirow{6}{*}{radius} & ~ & exact & $n$~\cite{HolzerW12,PelegRT12} & $\sqrt{nD}$ & ~ & $n$ & $\sqrt[3]{nD^2}+\sqrt n$ \\
~ & ~ & $3/2-\varepsilon$ & $n$ & $\sqrt{nD}$ & ~ & $n$~\cite{AbboudCK16} & $\sqrt n+D$ \\
~ & ~ & $3/2$ & $\sqrt n+D$~\cite{AnconaCDEW20} & $\sqrt n+D$ & ~ & open & open \\
~ & weighted & exact & $n$~\cite{BernsteinN19} & $n$ & ~ & $n$ & $n^{2/3}$ \\
~ & weighted & $(1,3/2)$ & $n$ & \textcolor{red}{$\min\left\{n^{9/10}D^{3/10},n\right\}$ (This work)} & ~ & $n$ & \textcolor{red}{$n^{2/3}$ (This work)} \\
~ & weighted & $2$ & $\sqrt nD^{1/4}+D$~\cite{ChechikM20} & $\sqrt nD^{1/4}+D$ & ~ & open & open \\
\bottomrule[1.5pt]
\end{tabular}
\end{table} 

\section{Preliminaries}

\subsection{Graph Notations}
\label{sec:graph_notations}

Given a weighted graph $(G,w)$ where $G=(V,E)$ and $w:E\to\mathbb N^+$.
The length of a path is defined to be the sum of weights of edges on it, and the distance between nodes $u$ and $v$, denoted by $d_{G,w}(u,v)$, is the least length over all paths between them.
The eccentricity of a node $u$ is denoted by $e_{G,w}(u)=\max_{v\in V}d_{G,w}(u,v)$.
The radius of weighted graph $(G,w)$, denoted by $R_{G,w}$, is the minimum eccentricity over all nodes, i.e., $R_{G,w}=\min_{u\in V}e_{G,w}(u)$, while the diameter of $(G,w)$, denoted by $D_{G,w}$, is the maximum eccentricity of nodes, or equally, the maximum distance between any two nodes, i.e., $D_{G,w}=\max_{u\in V}e_{G,w}(u)=\max_{u,v\in V}d_{G,w}(u,v)$.
The unweighted diameter of graph $G$ is denoted by $D_G=D_{G,w^\star}$ where $w^\star(e)=1$ for all $e\in E$, which is an essential parameter when $G$ represents the underlying graph of a distributed network.

\subsection{CONGEST Model}

In the classical CONGEST model, the communication network is a graph $G=(V,E)$ with $n$ nodes, and every node is assigned with a unique identifier.
Each node represents a processor with unlimited computational power, i.e., the consumption of any local computation in a single processor is ignored.
Each edge connecting two nodes represents a communication channel with $B=O(\log n)$ bits of bandwidth.
In this article, we further consider the weighted graph $(G,w)$ as underlying network, where the weight of each edge is initially known to both of its endpoints.

For quantum version of the CONGEST model defined in \cite{ElkinKNP14}, adjacent nodes are allowed to exchange \textit{qubits} (quantum bits), i.e., the classical channels are now quantum channels with the same bandwidth $B=O(\log n)$.
Each node can locally do some quantum computation, and distinct nodes may own qubits with entanglement.
In this paper, we assume that initially all nodes do not share any entanglement, but the nodes can, for example, locally create a pair of entangled qubits, and send one to others.

For both classical and quantum CONGEST models, the algorithm is implemented round by round in a synchronous manner.
In each round, each node sends/receives a message of $O(\log n)$ (qu)bits to/from each neighbor, and then does local computation according to local knowledge.
The algorithm halts when all nodes halt, and at the end of the algorithm, each node has its own output. We say an algorithm computes the diameter/radius if all nodes output the correct answer.
The round complexity of an algorithm in this model is defined to be the number of communication rounds needed.
And the round complexity of a distributed problem is the least round complexity of any algorithm solving it.
Our focus here is the distance problems, mainly the computation of diameter and radius mentioned in Section~\ref{sec:graph_notations}.

\subsection{Server Model}

The Sever model  is a variant of the two-party communication model, which was introduced by Elkin et al.~\cite{ElkinKNP14} to prove lower bounds in the CONGEST model.
There are three players in the Server model: Alice, Bob, and the server.
Alice and Bob receive the inputs $x$ and $y$ respectively, and want to compute $F(x,y)$ for some function $F$.
The server receives no input.
Alice and Bob can exchange messages with the server.
The catch here is that the server can send messages for free.
Thus, the communication complexity counts only messages sent by Alice and Bob.
Note that Alice and Bob can talk to each other by considering the server as a communication channel, so any protocol in the traditional two-party communication model can be implemented in the Server model with the same complexity.

For a two-argument function $F$ and $0\le\varepsilon<1$, we let $Q^{sv}_\varepsilon(F)$ denote the communication complexity (in the quantum setting) of computing $F$ where for any inputs $x,y$, the algorithm must output $F(x,y)$ with probability at least $1-\varepsilon$.
For Boolean function $f:\{0,1\}^n\rightarrow\mathbb{R}$ and $0\le\varepsilon<1$, the $\varepsilon$-approximate degree of $f$, denoted by $\text{deg}_\varepsilon(f)$, is the smallest degree of any polynomial $p$ that $\varepsilon$-represents $f$, i.e., $|p(x)-f(x)|\le\varepsilon$ for any input $x\in\{0,1\}^n$.

\section{Algorithm}

We first introduce the framework of distributed quantum optimization in \cite{GallM18}.
Given function  $f:X\to\mathbb Z$, where $X$ is a finite set, let $G=(V,E)$ be a network with a pre-defined node $\text{leader}\in V$.
We write $\ket\phi_v$ to denote a state in the memory space of node $v$.
A specific register $\ket\cdot_I$ called \textit{internal} and the control of the algorithm are centralized by the node leader.
Assume that the following three quantum procedures are given as black boxes.
\begin{itemize}
\item {\bf Initialization:} Prepare an initial state $\ket0_I\ket{\text{init}}$ with some pre-computed information $\ket{\text{init}}$.
\item {\bf Setup:} Produce a superposition from the initial state:
$$\ket0_I\ket{\text{init}}\mapsto\sum_{x\in X}\alpha_x\ket x_I\ket{\text{data}(x)}\ket{\text{init}},$$
where the $\alpha_x$'s are arbitrary amplitudes and $\text{data}(x)$ are information depending on $x$.
\item {\bf Evaluation:} Perform the transformation
$$\ket{x,0}_I\ket{\text{data}(x)}\ket{\text{init}}\mapsto\ket{x,f(x)}_I\ket{\text{data}(x)}\ket{\text{init}}.$$
\end{itemize}

The following lemma provides an algorithm to search $x\in X$ with high value $f(x)$ given the three procedures above.

\begin{lemma}[Theorem 2.4 in \cite{GallM18}\footnote{Although Le~Gall and Magniez write a slightly weaker statement, the lemma we claim here can be proven by the same argument in \cite{GallM18}.}]
Assume that {\rm Initialization} can be implemented within $T_0$ rounds in the quantum CONGEST model, and that unitary operators {\rm Setup} and {\rm Evaluation} and their inverses can be implemented within $T$ rounds.
Let $\rho>0$ be such that $\sum_{x\in X:f(x)\ge M}|\alpha_x|^2\ge\rho$ where $M$ is unknown to all nodes.
Then, for any $\delta>0$, the node {\rm leader} can find, with probability at least $1-\delta$, some element $x$ such that $f(x)\ge M$, in $T_0+O\left(\sqrt{\log(1/\delta)/\rho}\right)\times T$ rounds.
\label{lem:optimization}
\end{lemma}

The three procedures will be described as deterministic or randomized procedures that combine the subroutines provided by Nanongkai~\cite{Nanongkai14STOC} (also presented in Appendix~\ref{sec:toolkits}).
They can be quantized using the standard technique \cite{Bennett89}, with potentially additional garbage whose size is of the same order as the initial memory space.

Given a weighted graph $(G,w)$ where $G=(V,E)$ is a network and $w:E\to\mathbb N^+$, we show a quantum algorithm approximating $D_{G,w}$ and $R_{G,w}$ by proving Theorem~\ref{thm:upper_bound}.
We only show the algorithm approximating the diameter.
The proof for radius is basically the same except that it finds the minimum (approximate) eccentricity instead of the maximum one.

We choose the parameters throughout this section.
\begin{equation}
\label{eqn:alg_parameters}
\varepsilon=1/\log n,r=n^{2/5}D_G^{-1/5},\ell=n\log n/r,k=\sqrt{D_G}.
\end{equation}
As mentioned in Section~\ref{sec:results}, finding a node with maximum eccentricity among all nodes by directly applying a quantum search algorithm can hardly be done in $o(n)$ rounds.
We instead try to find a vertex set containing a node with maximum approximate eccentricity among $n$ vertex sets $S_1,\cdots,S_n$, and then search such a node in this vertex set.
Each set $S_i$ for $i\in[1,n]$ is sampled by having each node $v\in V$ join it independely with probability $r/n$.
For such a random set and a node $s$ in it, Nanongkai showed in \cite{Nanongkai14STOC} an efficient classical procedure to approximate its eccentricity (actually every node $v\in V$ can know an approximation of the distance from $s$ to $v$).

\subsection{Computation of Approximate Eccentricity}
\label{sec:approx_ecc}

For convenience, we need to introduce several graph notations.
Given a weighted graph $(G,w)$, the hop distance between nodes $u$ and $v$, denoted by $h_{G,w}(u,v)$, is the minimum number of edges over all shortest paths between them.
The hop diameter of the weighted graph, denoted by $H_{G,w}$, is the maximum hop distance between any two nodes.
For $\ell>0$, the $\ell$-hop distance between $u$ and $v$, denoted by $d^\ell_{G,w}(u,v)$, is the least length over all paths between them containing at most $\ell$ edges.
Note that $d^\ell_{G,w}(u,v)=d_{G,w}(u,v)$ when $h_{G,w}(u,v)\le\ell$.

In general, Nanongkai~\cite{Nanongkai14STOC} would approximate the bounded-hop distance, and sample a random set of key nodes as skeleton.
Then it could approximate the distance from any key node $s$ to any node $v$ since, with high probability, any shortest path from $s$ to $v$ can be partitioned into bounded-hop shortest paths between key nodes, along with a tail path from some key node to $v$, as long as the number of key nodes is sufficiently large.

Here we only list the necessary definitions of approximate bounded-hop distance, approximate distance, and approximate eccentricity.
We claim that these are good approximations.
The algorithms evaluating these quantities are presented in Appendix~\ref{sec:toolkits}, and the detailed proof should be found in \cite[arXiv version]{Nanongkai14STOC}.
Note that we are given a weighted graph $(G,w)$ where $G=(V,E)$ and $w:E\to\mathbb N^+$.

\begin{lemma}[Theorem 3.3 in \cite{Nanongkai14STOC}]
Given an integer $\ell>0$.
For integer $i\ge0$, define $w_i:E\to\mathbb N^+$ where $w_i(e)=\left\lceil\frac{2\ell w(e)}{\varepsilon\cdot2^i}\right\rceil$ for $e\in E$.
For any $u,v\in V$, the approximate bounded-hop distance is defined as
$$\widetilde d^\ell_{G,w}(u,v)=\min_i\left\{d_{G,w_i}(u,v)\cdot\frac{\varepsilon\cdot2^i}{2\ell}:d_{G,w_i}(u,v)\le\left(1+\frac2\varepsilon\right)\ell\right\}.$$
Then $d_{G,w}(u,v)\le\widetilde d^\ell_{G,w}(u,v)\le(1+\varepsilon)d^\ell_{G,w}(u,v)$.
\label{lem:bounded_hop_distance}
\end{lemma}

\begin{lemma}[Theorem 4.2 in \cite{Nanongkai14STOC}]
Given a vertex set $S\subseteq V$.
Let the weighted complete graph $\left(G'_S,w'_S\right)$ be such that
$$
\begin{aligned}
& G'_S=\left(S,\tbinom{S}{2}\right),w'_S:\tbinom{S}{2}\to\mathbb N^+, \\
& w'_S(\{u,v\})=\widetilde d^\ell_{G,w}(u,v),\forall\{u,v\}\in \tbinom{S}{2}.
\end{aligned}
$$
For node $v\in S$, let $N^k_S(v)$ be the set of the $k$ nodes with the least distance from $v$ on $(G'_S,w'_S)$.
And let the weighted complete graph $\left(G''_S,w''_S\right)$ be such that
$$
\begin{aligned}
& G''_S=\left(S,\tbinom{S}{2}\right),w''_S:\tbinom{S}{2}\to\mathbb N^+, \\
& w''_S(\{u,v\})=
\begin{cases}
d_{G'_S,w'_S}(u,v),	& u\in N^k_S(v)\text{ or }v\in N^k_S(u) \\
w'_S(\{u,v\}), 	& \text{otherwise}
\end{cases},\forall\{u,v\}\in \tbinom{S}{2}.
\end{aligned}
$$
For any $s\in S$ and $v\in V$, the approximate distance is defined as
$$\widetilde d_{G,w,S}(s,v)=\min_{u\in S}\left\{\widetilde d^{4|S|/k}_{G''_S,w''_S}(s,u)+\widetilde d^\ell_{G,w}(u,v)\right\}.$$
If $\ell=n\log n/r$ and $S$ is sampled by having each node $v\in V$ join it independetly with probability $r/n$, then $d_{G,w}(s,v)\le\widetilde d_{G,w,S}(s,v)\le(1+\varepsilon)^2d_{G,w}(s,v)$ for all $s\in S$ and $v\in V$, with probability at least $1-2^{-cn}$, for some constant $c>0$ and sufficiently large $n$.
\label{lem:approx_distance}
\end{lemma}

\noindent
{\bf Remark.} We briefly explain why $\widetilde d_{G,w,S}(\cdot)$ is a good approximation.
By the choice of $\ell$ and $S$, Lemma 4.3 in \cite{Nanongkai14STOC} says that, with high probability, any $s\in S,v\in V$ and shortest path $\left(s\leadsto v\right)$ on $(G,w)$ is of the form $\left(s=s_1\leadsto\cdots\leadsto s_m=u\leadsto v\right)$ such that $s_i\in S$ for $i\in[1,m]$, $h_{G,w}(s_{i-1},s_i)\le\ell$ for $i\in[2,m]$, and $h_{G,w}(u,v)\le\ell$.
Apparently $\widetilde d_{G,w,S}\ge d_{G,w}(s,v)$.
On the other side,
\begin{equation*}
\begin{split}
\widetilde d_{G,w,S}(s,v) & =\widetilde d^{4|S|/k}_{G''_S,w''_S}(s,u)+\widetilde d^\ell_{G,w}(u,v) \\
& \le(1+\varepsilon)d^{4|S|/k}_{G''_S,w''_S}(s,u)+\widetilde d^\ell_{G,w}(u,v) \\
& =(1+\varepsilon)d_{G''_S,w''_S}(s,u)+\widetilde d^\ell_{G,w}(u,v) \\
& \le(1+\varepsilon)\sum_{i=2}^mw'_S(\{s_{i-1},s_i\})+\widetilde d^\ell_{G,w}(u,v) \\
& \le(1+\varepsilon)\sum_{i=2}^m\widetilde d^\ell_{G,w}(s_{i-1},s_i)+\widetilde d^\ell_{G,w}(u,v) \\
& \le(1+\varepsilon)^2\left(\sum_{i=2}^md^\ell_{G,w}(s_{i-1},s_i)+d^\ell_{G,w}(u,v)\right) \\
& =(1+\varepsilon)^2\left(\sum_{i=2}^md_{G,w}(s_{i-1},s_i)+d_{G,w}(u,v)\right) \\
& =(1+\varepsilon)^2d_{G,w}(s,v).
\end{split}
\end{equation*}
The second and sixth lines are due to Lemma~\ref{lem:bounded_hop_distance}.
The third line is due to Theorem 3.10 in \cite{Nanongkai14STOC}, which says that $H_{G''_S,w''_S}<4|S|/k$ since $(G''_S,w''_S)$ is the {\it $k$-shortcut graph} of $(G'_S,w'_S)$.

\bigskip

For $i\in[1,n]$, we rewrite $G''_{S_i},w''_{S_i},\widetilde d_{G,w,S_i}(\cdot)$ as $G''_i,w''_i,\widetilde d_{G,w,i}(\cdot)$ for short.
For any $s\in S_i$, the approximate eccentricity is defined as $\widetilde e_{G,w,i}(s)=\max_{v\in V}\widetilde d_{G,w,i}(s,v)$.
Define two good events:
\begin{itemize}
\item {\bf Good-Scale:} For all $i\in[1,n]$, $|S_i|=\Theta(r)$.
Besides, let $v^\star\in V$ be a node with maximum eccentricity, i.e., $e_{G,w}(v^\star)=D_{G,w}$, then $v^\star$ joins $\beta=\Theta(r)$ sets $S_{i_1},\cdots,S_{i_\beta}$.
\item {\bf Good-Approximation:} For all $i\in[1,n]$ and $s\in S_i,v\in V$, $d_{G,w}(s,v)\le\widetilde d_{G,w,i}(s,v)\le(1+\varepsilon)^2d_{G,w}(s,v)$, thus $e_{G,w}(s)\le\widetilde e_{G,w,i}(s)\le(1+\varepsilon)^2e_{G,w}(s)$.
\end{itemize}
By Chernoff inequality and a union bound, the event Good-Scale occurs with probability at least $1-1/\text{poly}(n)$.
By Lemma~\ref{lem:approx_distance} %A.5 in the full version~\cite{WuY22} %~\ref{lem:approx_distance}
and a union bound, the event Good-Approximation occurs with probability at least $1-1/\text{poly}(n)$.
Therefore, we can assume that the two events all happen in the following context.

\subsection{Quantization}

For each $i\in[1,n]$, we define $f_i:S_i\to\mathbb Z$ where $f_i(s)=\widetilde e_{G,w,i}(s)$ for $s\in S_i$, and  $f:[1,n]\to\mathbb Z$ where $f(i)=\max_{s\in S_i}f_i(s)$ for $i\in[1,n]$.

\begin{lemma}
The number of $i\in[1,n]$ satisfying $f(i)\ge D_{G,w}$ is $\Theta(r)$.
Moreover, $f(i)\le(1+\varepsilon)^2D_{G,w}$  for all $i\in[1,n]$.
\label{lem:approximation}
\end{lemma}

\begin{proof}
$$
\begin{aligned}
& f(i_j)=\max_{s\in S_{i_j}}\widetilde e_{G,w,i_j}(s)\ge\max_{s\in S_{i_j}}e_{G,w}(s)\ge e_{G,w}(v^\star)=D_{G,w}, & \forall j\in[1,\beta]; \\
& f(i)=\max_{s\in S_i}\widetilde e_{G,w,i}(s)\le\max_{s\in S_i}(1+\varepsilon)^2e_{G,w}(s)\le(1+\varepsilon)^2D_{G,w}, & \forall i\in[1,n].
\end{aligned}
$$
\end{proof}

\begin{lemma}
Given $i\in[1,n]$, there exists a quantum procedure performing the transformation
$$\bigotimes_{v\in V}\ket i_v\ket0_{{\rm leader}}\mapsto\bigotimes_{v\in V}\ket i_v\ket{f(i)}_{{\rm leader}}$$
in the quantum CONGEST model, and taking $\widetilde O\left(D_G+\frac n{\varepsilon\cdot r}+rk+\sqrt r\left(\frac r{\varepsilon\cdot k}\cdot D_G+r\right)\right)$ rounds, with probability at least $1-1/{\rm poly}(n)$.
\label{lem:evaluation}
\end{lemma}

\begin{proof}
We give the quantum procedure maximizing $f_i$ (thus evaluating $f(i)$) by following the framework of distributed quantum optimization:
\begin{itemize}
\item {\bf $\text{Initialization}_i$:} Perform the transformation
$$
\bigotimes_{v\in V}\ket i_v\mapsto\bigotimes_{v\in V}\ket i_v\ket0_I\ket{\text{init}_i},$$
where
$$\ket{\text{init}_i}=\bigotimes_{v\in V,u\in S_i}\left|\widetilde d^\ell_{G,w}(u,v)\right\rangle_v\ket{G''_i,w''_i},$$
and $d^\ell_{G,w}(u,v)$ is given in Lemma~\ref{lem:bounded_hop_distance}, $G''_i$ and $w''_i$ are given in lemma~\ref{lem:approx_distance}.

\item {\bf $\text{Setup}_i$:} Perform the transformation
$$\bigotimes_{v\in V}\ket i_v\ket0_I\ket{\text{init}_i}\mapsto\bigotimes_{v\in V}\ket i_v\left(\sum_{s\in S_i}\frac1{|S_i|}\ket s_I\ket{\text{data}_i(s)}\right)\ket{\text{init}_i},$$
where
$\ket{\text{data}_i(s)}=\bigotimes_{v\in V}\ket s_v\bigotimes_{v\in V,u\in S_i}\left|\widetilde d_{G''_i,w''_i}(s,u)\right\rangle_v$.

\item {\bf $\text{Evaluation}_i$:} Perform the transformation
$$\bigotimes_{v\in V}\ket i_v\left(\ket{s,0}_I\ket{\text{data}_i(s)}\right)\ket{\text{init}_i}\mapsto\bigotimes_{v\in V}\ket i_v\left(\ket{s,f_i(s)}_I\ket{\text{data}_i(s)}\right)\ket{\text{init}_i}.$$
\end{itemize}

We now analyze the round complexity:
\begin{itemize}
\item In $\widetilde O\left(D_G+\frac n{\varepsilon\cdot r}+r\right)$ rounds, each $v\in V$ can know $\widetilde d^\ell_{G,w}(u,v)$ for each $u\in S_i$, with high probability, due to Lemma~\ref{lem:bounded_hop_mssp}.
After that, the overlay network $\left(G''_i,w''_i\right)$ can be embedded in $O(D_G+rk)$ rounds due to Lemma~\ref{lem:embedding} (we say that the network $G=(V,E)$ embeds an overlay network $G'=(V',E')$ with a weight function $w':E'\to\mathbb N^+$ if $V'\subseteq V$ and for each $v\in V'$, it stores each $e\in E'$ incident to $v$ along with $w'(e)$ in the local memory).
Therefore, the procedure $\text{Initialization}_i$ can be implemented in $T_0=\widetilde O\left(D_G+\frac n{\varepsilon\cdot r}+rk\right)$ rounds.
\item The node leader can collect $S_i$ in $O(D_G+r)$ rounds.
It then prepares the quantum state $\sum_{s\in S_i}\frac1{|S_i|}\ket s_I$ and broadcasts to all nodes using CNOT copies, in $O(D_G)$ rounds.
Thus, the transformation
$$\bigotimes_{v\in V}\ket i_v\ket0_I\ket{\text{init}_i}\mapsto\bigotimes_{v\in V}\ket i_v\left(\sum_{s\in S_i}\frac1{|S_i|}\ket s_I\bigotimes_{v\in V}\ket s_v\right)\ket{\text{init}_i}$$
can be implemented in $O(D_G+r)$ rounds.
Besides, the transformation
$$\bigotimes_{v\in V}\ket i_v\bigotimes_{v\in V}\ket s_v\ket{\text{init}_i}\mapsto\bigotimes_{v\in V}\ket i_v\bigotimes_{v\in V}\ket s_v\bigotimes_{v\in V,u\in S_i}\left|\widetilde d^{4|S_i|/k}_{G''_i,w''_i}(s,u)\right\rangle_v\ket{\text{init}_i}$$
can be implemented in $T_1=\widetilde O\left(\frac r{\varepsilon\cdot k}\cdot D_G+r\right)$ rounds since Lemma~\ref{lem:sssp_on_overlay} implies that, after the overlay network $\left(G''_i,w''_i\right)$ is embedded, each $v\in V$ can know $\widetilde d^{4|S_i|/k}_{G''_i,w''_i}(s,u)$ for each $u\in S_i$ within $T_1$ rounds.
Therefore, the procedure $\text{Setup}_i$ can be implemented in $T_1=\widetilde O\left(\frac r{\varepsilon\cdot k}\cdot D_G+r\right)$ rounds.
\item For the procedure $\text{Evaluation}_i$, recall that $f_i(s)=\max_{v\in V}\widetilde d_{G,w,i}(s,v)$ where
$$\widetilde d_{G,w,i}(s,v)=\min_{u\in S_i}\left\{\widetilde d^{4|S_i|/k}_{G''_i,w''_i}(s,u)+\widetilde d^\ell_{G,w}(u,v)\right\}.$$
By definition, for any $v\in V$ and $u\in S_i$, $\widetilde d^{4|S_i|/k}_{G''_i,w''_i}(s,u)$ and $\widetilde d^\ell_{G,w}(u,v)$ have been stored in the local memory of $v$, i.e., $\ket\cdot_v$.
Thus, each $v\in V$ can locally compute $\widetilde d_{G,w,S_i}(s,v)$, and the node leader can compute the maximum by converge-casting in $O(D_G)$ rounds.
So the procedure $\text{Evaluation}_i$ can be implemented in $T_2=O(D_G)$ rounds.
\end{itemize}
By Lemma~\ref{lem:optimization}, there exists a quantum procedure maximizing $f_i$ in $\widetilde O(T_0+\sqrt r(T_1+T_2))$ rounds with high probability.
\end{proof}

\begin{proof}[Proof of Theorem~\ref{thm:upper_bound}]
We give a quantum procedure maximizing $f$ also by following the framework of distributed quantum optimization:
\begin{itemize}
\item \textbf{Initialization} is a classical procedure which samples vertex sets $S_1,\cdots,S_n$, and $\ket{\text{init}}$ represents the corresponding classical information.
\item \textbf{Setup:} Perform the transformation
$$\ket0_I\ket{\text{init}}\mapsto\sum_{i=1}^n\frac1n\ket i_I\ket{\text{data}(i)}\ket{\text{init}},$$
where $\ket{\text{data}(i)}=\bigotimes_{v\in V}\ket i_v$.
\item \textbf{Evaluation:} Perform the transformation
$$\ket{i,0}_I\ket{\text{data}(i)}\ket{\text{init}}\mapsto\ket{i,f(i)}_I\ket{\text{data}(i)}\ket{\text{init}}.$$
\end{itemize}

We now analyze the round complexity:
\begin{itemize}
\item $S_1,\cdots,S_n$ are sampled locally in parallel, and the procedure Initialization is free, i.e., $T_0=0$.
\item The node leader prepares the quantum state $\sum_{i=1}^n\frac1n\ket i_I$ and broadcast using CNOT copies to all nodes.
Therefore, the procedure Setup can be implemented in $T_1=O(D_G)$ rounds.
\item The procedure Evaluation can be of $T_2=\widetilde O\left(D_G+\frac n{\varepsilon\cdot r}+rk+\sqrt r\left(\frac r{\varepsilon\cdot k}\cdot D_G+r\right)\right)$ rounds by Lemma~\ref{lem:evaluation}.
\end{itemize}

By Lemma~\ref{lem:optimization} and Lemma~\ref{lem:approximation}, there exists a quantum procedure that find, with high probability, some $i\in[1,n]$ such that $D_{G,w}\le f(i)\le(1+\varepsilon)^2D_{G,w}$, in
$$\widetilde O(T_0+\sqrt{n/r}(T_1+T_2))=\widetilde O\left(\sqrt{n/r}\left(D_G+\frac n{\varepsilon\cdot r}+rk+\sqrt r\left(\frac r{\varepsilon\cdot k}\cdot D_G+r\right)\right)\right)$$
rounds.
By the choice of the parameters in Eq.~\eqref{eqn:alg_parameters}, Theorem~\ref{thm:upper_bound} follows.
\end{proof}

\section{Lower Bound}

To prove the lower bound on the round complexity of approximating (weighted) diameter in the quantum CONGEST model, we combine the reduction in~\cite{Elkin06,SarmaHKKNPPW12,ElkinKNP14} and the graph gadget in~\cite{AbboudCK16}.

\subsection{Reduction from Server Model}
\label{sec:reduction}

We briefly outline the reduction introduced by Elkin et al.~\cite{Elkin06,SarmaHKKNPPW12,ElkinKNP14} from the Server model to prove the hardness of certain graph problems such as diameter and radius.
We will introduce a distributed network $G=(V,E)$ and embed a certain two-argument function $F:\{0,1\}^k\times\{0,1\}^k\to\{0,1\}$ into the network by showing that if the instance on the network $G$ has a low round-complexity protocol in the quantum CONGEST model, then there exists a low communication complexity protocol for $F$ in the quantum Server model.
Thus, the hardness of diameters and radius in the quantum CONGEST model is reduced to proving the lower bounds the communication complexity in the quantum Server model.

\begin{figure}[ht]
\centering
\includegraphics[width=0.5\linewidth]{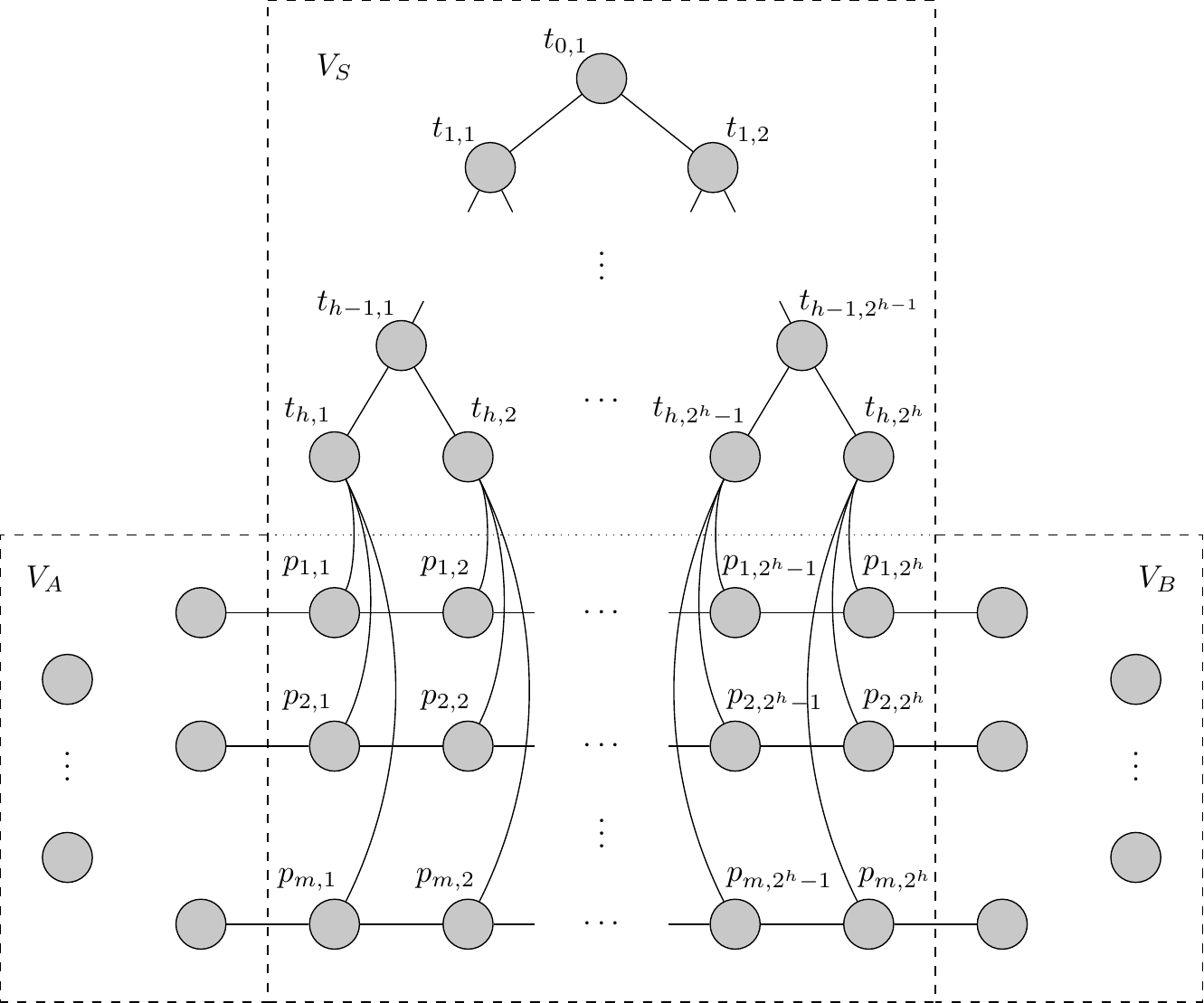}
\caption{An example of constructed graph $G$.}
\label{fig:example}
\end{figure}

The network $G=(V,E)$ is depicted by Figure~\ref{fig:example} where $V=V_S\uplus V_A\uplus V_B$ and $E=E_S\uplus E_A\uplus E_B\uplus E'$.
We use $G[U]$ to denote the subgraph induced by vertex set $U\subseteq V$, then $E_S,E_A,E_B$ are the edges in $G[V_S],G[V_A],G[V_B]$ respectively.
And $E'$ denotes the edges between $V_S$ and $V_A\uplus V_B$.

$G[V_S]$ includes a full binary tree of height $h$ and $m$ disjoint paths of length $2^h-1$.
Each of the $2^h$ leaves of the binary tree is connected to the nodes on the paths as depicted in Figure~\ref{fig:example}.
Suppose nodes of depth $i$ on the tree are $t_{i,1},\cdots,t_{i,2^i}$ and nodes on the $i$-th path are $p_{i,1},\ldots,p_{i,2^h}$ from left to right.
Then $t_{h,1},\ldots, t_{h,2^h}$ are the leaves of the binary tree in $G[V_S]$.
For each $i\in[1,m]$ and $j\in[1,2^h]$, there is an edge between $t_{h,j}$ and $p_{i,j}$.
Thus,
$$
\begin{aligned}
V_S & =\left\{t_{i,j}:i\in[0,h],j\in[1,2^i]\right\} \\
& \uplus\left\{p_{i,j}:i\in[1,m],j\in[1,2^i]\right\}, \\
E_S & =\left\{\{t_{i,j},t_{i-1,\lceil j/2\rceil}\}:i\in[1,h],j\in[1,2^i]\right\} \\
& \uplus\left\{\{p_{i,j},p_{i,j-1}\}:i\in[1,m],j\in[2,2^m]\right\} \\
& \uplus\left\{\{t_{h,j},p_{i,j}\}:i\in[1,m],j\in[1,2^m]\right\}.
\end{aligned}
$$
$V_A$ contains at least $m$ nodes, each of which is connected to $p_{i,1}$ for $1\leq i\leq m$.
$V_B$ contains at least $m$ nodes, each of which is connected to $p_{i,2^h}$ for $1\leq i\leq m$.
Those $2m$ edges are contained in $E'$.
The subgraphs $G[V_A]$ and $G[V_B]$ are decided by Alice's input and Bob's input, respectively.

\bigskip

The following lemma gives an efficient simulation of algorithms on network $G$ by the protocols in the quantum Server model.

\begin{lemma}[Quantum Simulation Lemma]
Suppose Alice and Bob are given $(V_A,E_A)$ and $(V_B,E_B)$, respectively.
For any $T$-round ($T<2^h/2$) distributed algorithm on network $G$ described above, there exists a communication protocol for Alice and Bob in the quantum Server model to simulate the algorithm with communication complexity $O(T\cdot h\cdot B)$,  where $B$ denotes the bandwidth in the CONGEST model.
\label{lem:simulation}
\end{lemma}

\begin{proof}
The proof of Lemma~\ref{lem:simulation} follows closely with the proof in \cite[Proof of Theorem 3.5]{ElkinKNP14}.
The protocol we will construct simulates the distributed algorithm round by round.
Thus, it also has $T<2^h/2$ rounds of communication.
In the beginning, the server simulates all the nodes in $V_S$ which are independent of Alice and Bob's inputs.
And in the end of the $r$-th round, the server simulates $p_{i,1+r},\cdots,p_{i,2^h-r}$ on the $i$-th path and nodes $t_{h,1+r},\cdots,t_{h,2^h-r}$ along with their ancestors on the binary tree, while Alice simulates the nodes on the left side and Bob simulates on the right side.
More formally, in the end of the $r$-th round, the server simulates
$$\left\{p_{i,j}:i\in[1,m],j\in[1+r,2^h-r]\right\}\cup\left\{t_{i,j}:i\in[0,h],j\in\left[\left\lceil(1+r)/2^{h-i}\right\rceil,\left\lceil(2^h-r)/2^{h-i}\right\rceil\right]\right\};$$
Alice simulates
$$V_A\cup\left\{p_{i,j}:i\in[1,m],j\in[1,1+r)\right\}\cup\left\{t_{i,j}:i\in[0,h],j\in\left[1,\left\lceil(1+r)/2^{h-i}\right\rceil\right)\right\};$$
Bob simulates
$$V_B\cup\left\{p_{i,j}:i\in[1,m],j\in(2^h-r,2^h]\right\}\cup\left\{t_{i,j}:i\in[0,h],j\in\left(\left\lceil(2^h-r)/2^{h-i}\right\rceil,2^i\right]\right\}.$$

We describe the simulation of the computation and communication of a processor $v$ in the $r$-th round, and count the total communication complexity.
\begin{itemize}
\item If $v$ is owned by Alice or the server in the $(r-1)$-th round and will be owned by Alice in the $r$-th round, Alice needs the local information of $v$ in the $(r-1)$-th round and messages from $\Gamma(v)$ (neighbours of $v$) to $v$ in the $r$-th round, which can be obtained by local computation and communication from the server to Alice since $T<2^h/2$, which implies that each of $v$ and nodes in $\Gamma(v)$ is owned by either Alice or the server in the $(r-1)$-th round for $r\leq T$.
So in this case, we only need communication from the server to Alice in the Server model.
This part will not be counted to complexity by definition.
\item If $v$ is owned by Bob or the server in the $(r-1)$-th round and will be owned by Bob in the $r$-th round, no communication will be counted to complexity by the same argument as mentioned above.
\item If $v$ is owned by the server in both the $(r-1)$-th round and the $r$-th round,
the server needs the messages from $\Gamma(v)$ to $v$.
For each node $u\in\Gamma(v)$ owned by Alice or Bob in the $(r-1)$-th round, Alice or Bob will simulates the local computation of $u$ in the $r$-th round, and send the message to the server.
\begin{itemize}
\item If $v$ is on the paths on $V_S$, none of $\Gamma(v)$ is owned by Alice and Bob in the $(r-1)$-th round.
\item If $v$ is on the binary tree, node $u\in\Gamma(v)$ is owned by Alice in the $(r-1)$-th round only if all nodes of the same depth with $v$, meanwhile on the left side of $v$, are not owned by the server in the $r$-th round, and $u$ is the left-child of $v$.
Similarly, node $u\in\Gamma(v)$ is owned by Bob in the $(r-1)$-th round only if all nodes of the same depth with $v$, meanwhile on the right side of $v$, are not owned by the server in the $r$-th round, and $u$ is the right-child of $v$.
In the $r$-th round, there are at most $2h$ such $(u,v)$ in total.
\end{itemize}
\end{itemize}
Hence, a total of $O(T\cdot h)$ messages, each of size $O(B)$, are sent from Alice or Bob to the server.
\end{proof}

\subsection{Hardness of Approximating Diameter}

we will use $G$ constructed above as a gadget to prove a lower bound on round complexity of approximating weighted diameter in the quantum CONGEST model.
The specific graph depicted in Figure~\ref{fig:diameter} will contain $n=\left(2^{h+1}-1\right)+\left(2s+\ell\right)\left(2^h+2\right)+2\cdot2^s$ nodes, where parameters $h,s,\ell$ are chosen as follows throughout this section.
\begin{equation}
\label{eqn:paramters}
h\text{ is some even number},s=3h/2,\ell=2^{s-h}.
\end{equation}
This choice makes $2^h=\widetilde\Theta(n^{2/3}),2^s=\widetilde\Theta(n)$ and $\ell=\widetilde\Theta(n^{1/3})$.

\begin{figure*}[ht]
\centering
\includegraphics[width=\linewidth]{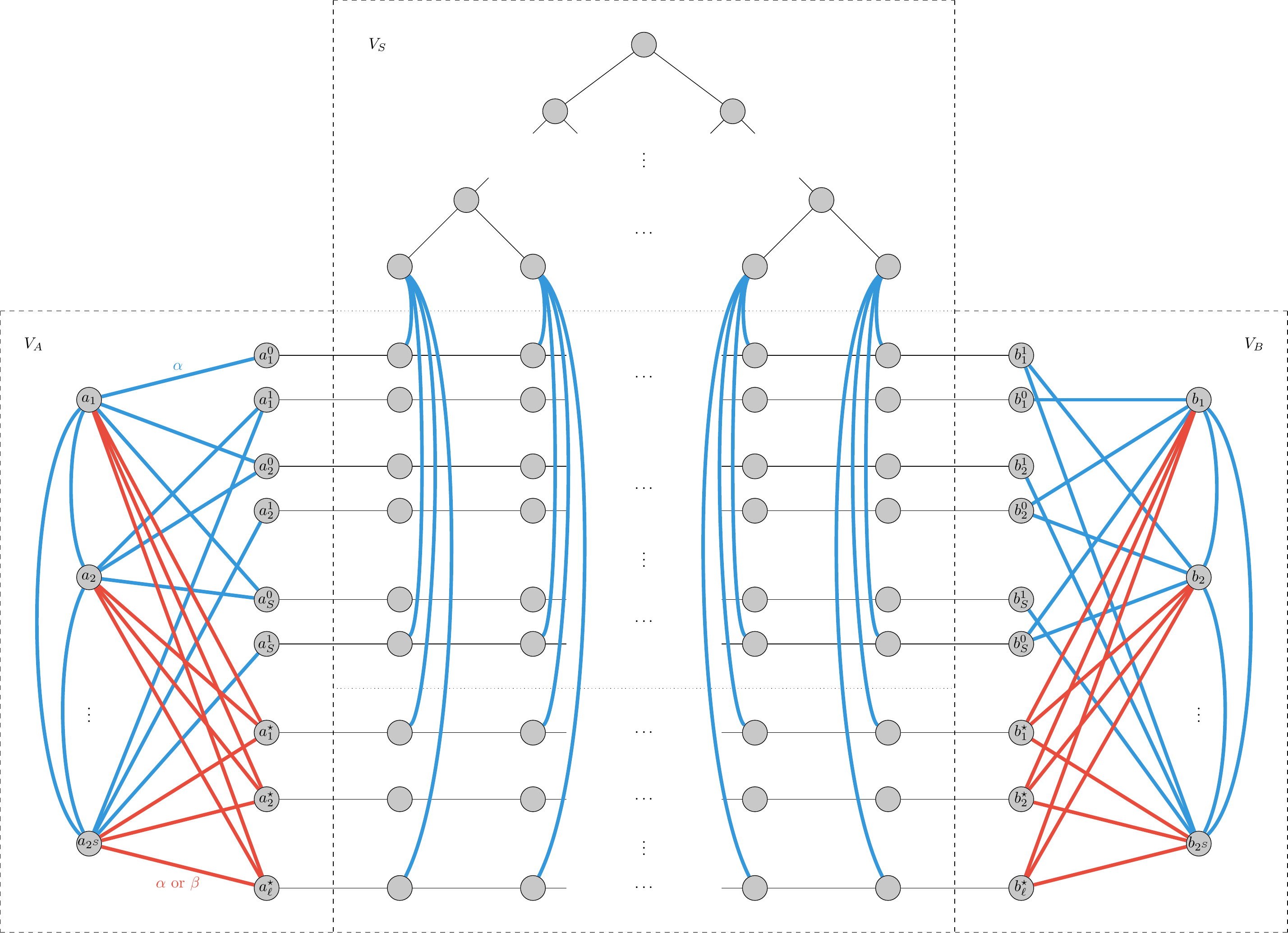}
\caption{
Graph $G$ for approximating diameter.
The black edges are of weight $1$; the blue edges are of weight $\alpha$; and weights of red edges are determined by inputs $x,y$, i.e, $w(\{a_i,a^\star_j\})=\alpha$ if $x_{i,j}=1$ and $w(\{a_i,a^\star_j\})=\beta$ if $x_{i,j}=0$, and $w(\{b_i,b^\star_j\})=\alpha$ if $y_{i,j}=1$ and $w(\{b_i,b^\star_j\})=\beta$ if $y_{i,j}=0$, for $i\in[1,2^s]$ and $j\in[1,\ell]$.
}
\label{fig:diameter}
\end{figure*}

\begin{theorem}[Restated]
For any constant $\varepsilon\in(0,\frac12]$, any algorithm, with probability at least $\frac{11}{12}$, computing a $(\frac32-\varepsilon)$-approximation of the weighted diameter in the quantum CONGEST model requires $\Omega\left(\frac{n^{2/3}}{\log^2n}\right)$ rounds, even when the unweighted diameter is $\Theta(\log n)$, where $n$ denotes the number of nodes.
\label{thm:diameter}
\end{theorem}

On network $G=(V,E)$ described in Section~\ref{sec:reduction}, we specify $G[V_A]$ and $G[V_B]$.
Let
\begin{equation*}
\begin{aligned}
V_A & =\left\{a_1,\cdots,a_{2^s}\right\}\uplus\left\{a^0_1,a^1_1,\cdots,a^0_s,a^1_s\right\}\uplus\left\{a^\star_1,\cdots,a^\star_\ell\right\}, \\%\text{ and} \\
V_B & =\left\{b_1,\cdots,b_{2^s}\right\}\uplus\left\{b^0_1,b^1_1,\cdots,b^0_s,b^1_s\right\}\uplus\left\{b^\star_1,\cdots,b^\star_\ell\right\}.
\end{aligned}
\end{equation*}
The edges $E_A, E_B$ and $E'$ are specified as follows.
$$
\begin{aligned}
E_A & =\left\{\{a_i,a^{\text{bin}(i,j)}_j\}:i\in[1,2^s],j\in[1,s]\right\} \\
& \uplus\left\{\{a_i,a^\star_j\}:i\in[1,2^s],j\in[1,\ell]\right\} \\
& \uplus\left\{\{a_i,a_j\}:i,j\in[1,2^s],i\ne j\right\}, \\
E_B & =\left\{\{b_i,b^{\text{bin}(i,j)}_j\}:i\in[1,2^s],j\in[1,s]\right\} \\
& \uplus\left\{\{b_i,b^\star_j\}:i\in[1,2^s],j\in[1,\ell]\right\} \\
& \uplus\left\{\{b_i,b_j\}:i,j\in[1,2^s],i\ne j\right\}, \\
E' & =\left\{\{a^0_i,p_{2i-1,1}\},\{b^1_i,p_{2i-1,2^h}\}:i\in[1,s]\right\} \\
& \uplus\left\{\{a^1_i,p_{2i,1}\},\{b^0_i,p_{2i,2^h}\}:i\in[1,s]\right\} \\
& \uplus\left\{\{a^\star_i,p_{2s+i,1}\},\{b^\star_i,p_{2s+i,2^h}\}:i\in[1,\ell]\right\},
\end{aligned}
$$
where  $\text{bin}(i,j)$ denote the $j$-th bit in binary expression of integer $i-1$.

The node pairs $(a^0_i, p_{2i-1,1})$, $(a^1_i,p_{2i,1})$, $(b^0_i,p_{2i,2^h})$, $(b^1_i,p_{2i-1,2^h})$ for $1\le i\le s$, and $(a^\star_j,p_{2s+j,1})$, $(b^\star_j,p_{2s+j,2^h})$ for $1\le j\le\ell$ are connected.
For each $i\in[1,2^s]$, $a_i$ is connected to $a^{\text{bin}(i,j)}_j$ for each $j\in[1,s]$, and $a_i$ is connected to $a^\star_j$ for each $j\in[1,\ell]$.
Moreover, $G[\{a_1,\cdots,a_{2^s}\}]$ is a clique.
The edges in $G[V_B]$ are linked in the same way as the edges in $G[V_A]$.

The weights of the edges are specified as follows, which are also depicted in Figure~\ref{fig:diameter}.
\begin{itemize}
\item The edges on the binary tree and the edges on the $2s+\ell$ paths (including the endpoints in $V_A$ and $V_B$) are of weight $1$ (the black edges in Figure~\ref{fig:diameter}).
\item Recall that Alice and Bob receive inputs $x,y\in\{0,1\}^{2^s\cdot\ell}$ respectively.
$x$ and $y$ are indexed by $x_{i,j}$ and $y_{i,j}$ for $i\in[1,2^s],j\in[1,\ell]$ where $s$ and $\ell$ are given in Eq.~\eqref{eqn:paramters}.
For each $i\in[1,2^s],j\in[1,\ell]$, $w(\{a_i,a^\star_j\})=\alpha$ if $x_{i,j}=1$ and $w(\{a_i,a^\star_j\})=\beta$ if $x_{i,j}=0$ ($\alpha<\beta$); weights of edges between $\{b_1,\cdots,b_{2^s}\}$ and $\{b^\star_1,\cdots,b^\star_\ell\}$ are assigned according to $y$ in the same way (the red edges in Figure~\ref{fig:diameter}).
\item The edges between the binary tree and the $2s+\ell$ paths, those between $\{a_1,\cdots,a_{2^s}\}$ and $\{a^0_1,a^1_1,\cdots,a^0_s,a^1_s\}$, and those between $\{b_1,\cdots,b_{2^s}\}$ and $\{b^0_1,b^1_1,\cdots,b^0_s,b^1_s\}$ are of weight $\alpha$; weights of edges inside $G[\{a_1,\cdots,a_{2^s}\}]$ and $G[\{b_1,\cdots,b_{2^s}\}]$ are also $\alpha$ (the blue edges in Figure~\ref{fig:diameter}).
\end{itemize}

It is sufficient to analyze the diameter of graph after contracting all edges of weight $1$ due to the following lemma.
An edge is contracted if the two endpoints are merged to one node, and the adjacent edges of the two endpoints are incident to it.
If there are parallel edges after contraction,  we only keep the one with the lowest weight.

\begin{lemma}
Given a weighted graph $(G,w)$ where $G=(V,E)$ and $w:E\to\mathbb N^+$.
Let $G'$ be the graph after contracting all edges of weight $1$.
We have $D_{G',w}\le D_{G,w}\le D_{G',w}+n$ and $R_{G',w}\le R_{G,w}\le R_{G',w}+n$, where $n=|V|$.
\label{lem:contraction}
\end{lemma}
\begin{proof}
For any path $P$ in $G$, let $P'$ be the path in $G'$ obtained from $P$ after contraction.
Then
$$\text{length}(P')\leq\text{length}(P)\leq\text{length}(P')+n$$
as there are at most $n-1$ $1$-weight edges.
Thus we conclude the result.
\end{proof}

For inputs $x,y\in\{0,1\}^{2^s\cdot\ell}$ received by Alice and Bob, define
$$F(x,y)=\bigwedge_{i\in[1,2^s]}\left(\bigvee_{j\in[1,\ell]}\left(x_{i,j}\wedge y_{i,j}\right)\right),$$
i.e., $F=\text{AND}_{2^s}\circ(\text{OR}_\ell\circ\text{AND}^\ell_2)^{2^s}$.
We have the following lemma.

\begin{lemma}
$D_{G,w}\le\max\{2\alpha,\beta\}+n$ if $F(x,y)=1$, and $D_{G,w}\ge\min\{\alpha+\beta,3\alpha\}$ otherwise.
\label{lem:reduction}
\end{lemma}

\begin{proof}
The graph $G'$ after contraction is given in Figure~\ref{fig:contraction}.
The binary tree is contracted to node $t$.
The $2s+\ell$ paths are contracted to nodes $a^0_1,a^1_1,\cdots,a^0_s,a^1_s$ and $a^\star_1,\cdots,a^\star_\ell$ respectively.
Note that $b_i$ is connected to $a^{\text{bin}(i,j)\oplus1}_j$ for $i\in[1,2^s],j\in[1,s]$.
we list upper bounds of the distances between any two nodes $u$ and $v$ in $G'$ on Table~\ref{tab:distance} with the corresponding paths, except for the distance between $a_i$ and $b_i$ with $i\in[1,2^s]$.

\begin{figure}[ht]
\centering
\includegraphics[width=0.5\linewidth]{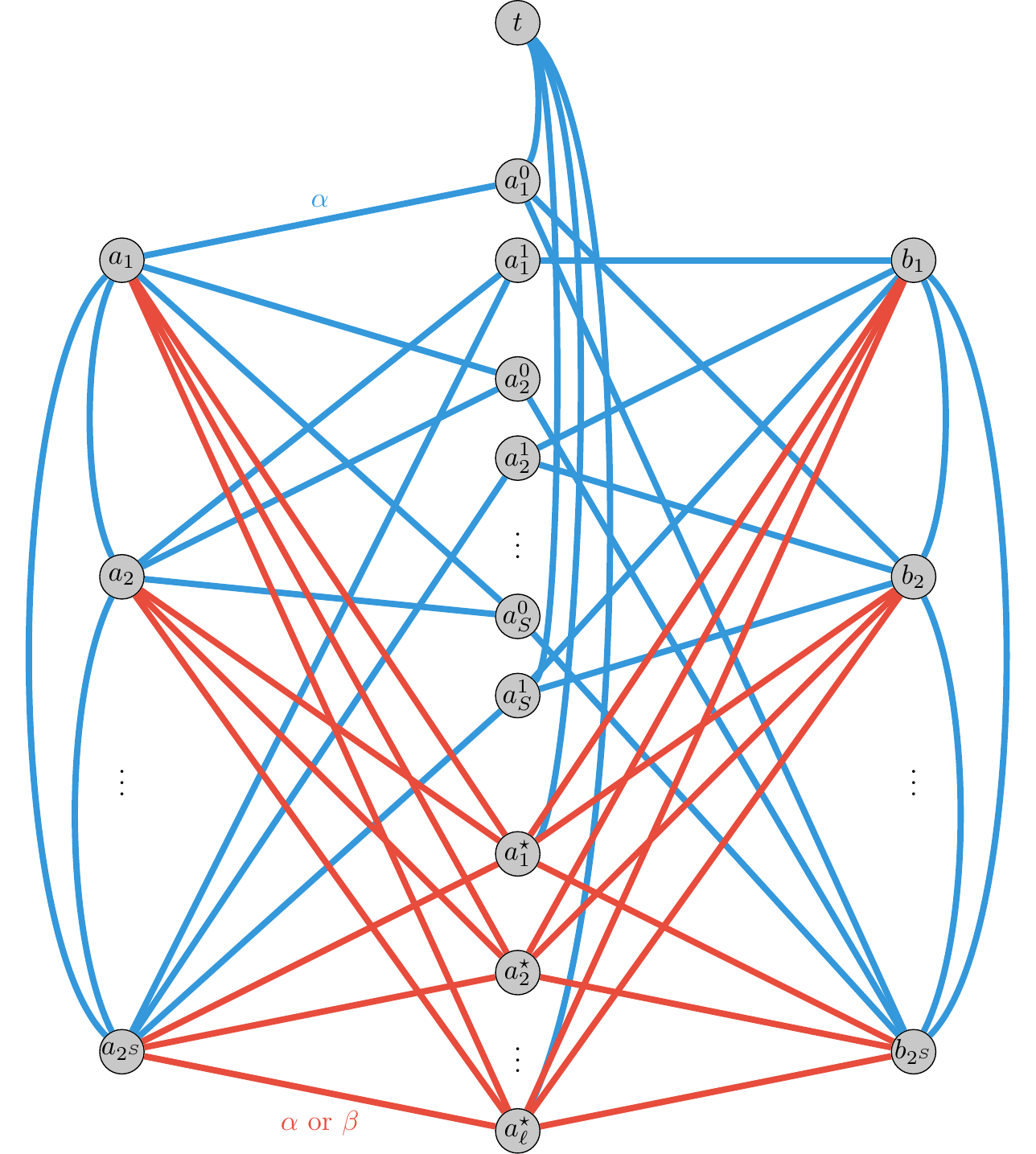}
\caption{
Graph $G'$ after contraction.
The distance between any pair of nodes, except $a_i$ and $b_i$ for $i\in[1,2^s]$, is at most $\max\{2\alpha,\beta\}$; and the distance between $a_i$ and $b_i$ is at most $2\alpha$ if there exists $j\in[1,\ell]$ such that $x_{i,j}=y_{i,j}=1$, otherwise it is at least $\min\{\alpha+\beta,3\alpha\}$.
Therefore, the diameter is at most $\max\{2\alpha,\beta\}$ if, for any $i\in[1,2^s]$, there exists $j\in[1,\ell]$ such that $x_{i,j}=y_{i,j}=1$, otherwise it is at least $\min\{\alpha+\beta,3\alpha\}$.
}
\label{fig:contraction}
\end{figure}

\begin{table*}[t]
\centering
\caption{
Distance between nodes in $G'$.
Let {\rm router} be any node in $\{a^0_1,a^1_1,\cdots,a^0_s,a^1_s,a^\star_1,\cdots,a^\star_\ell\}$.
${\rm adj}(i,j)$ denotes the integer after changing the $j$-th bit in binary expression of integer $i-1$, and ${\rm ind}(i,j)$ is the smallest $z\in[1,s]$ satisfying ${\rm bin}(i,z)\ne{\rm bin}(j,z)$.
}
\label{tab:distance}
\begin{tabular}{cccc}
\toprule[1.5pt]
$u$                                         & $v$                                               & $d_{G',w}(u,v)$   & Path \\
\midrule[1.5pt]
\multirow{3.6}{*}{$t$} 						& router                                            & $\le\alpha$       & $\left(t\to v\right)$ \\
											& $a_i$ ($i\in[1,2^s]$)                             & $\le2\alpha$      & $\left(t\to a^{\text{bin}(i,0)}_0\to a_i\right)$ \\
											& $b_i$ ($i\in[1,2^s]$)                             & $\le2\alpha$      & $\left(t\to a^{\text{bin}(i,0)\oplus1}_0\to b_i\right)$ \\
\midrule
\multirow{6.4}{*}{$a_i$ ($i\in[1,2^s]$)}    & $a_j$ ($j\ne i,j\in[1,2^s]$)                      & $\le\alpha$       & $\left(a_i\to a_j\right)$ \\
											& $a^{\text{bin}(i,j)}_j$ ($j\in[1,s]$)             & $\le\alpha$       & $\left(a_i\to a^{\text{bin}(i,j)}_j\right)$ \\
											& $a^{\text{bin}(i,j)\oplus1}_j$ ($j\in[1,s]$)      & $\le2\alpha$      & $\left(a_i\to a_{\text{adj}(i,j)}\to a^{\text{bin}(i,j)\oplus1}_j\right)$ \\
											& $b_j$ ($j\ne i,j\in[1,2^s]$)                      & $\le2\alpha$      & $\left(a_i\to a^{\text{bin}(i,\text{ind}(i,j))}_{\text{ind}(i,j)}\to b_j\right)$ \\
											& $a^\star_j$ ($j\in[1,\ell]$)                      & $\le\beta$        & $\left(a_i\to a^\star_j\right)$ \\
\midrule
\multirow{5.2}{*}{$b_i$ ($i\in[1,2^s]$)}    & $b_j$ ($j\ne i,j\in[1,2^s]$)                      & $\le\alpha$       & $\left(b_i\to b_j\right)$ \\
											& $a^{\text{bin}(i,j)\oplus1}_j$ ($j\in[1,s]$)      & $\le\alpha$       & $\left(b_i\to a^{\text{bin}(i,j)\oplus1}_j\right)$ \\
											& $a^{\text{bin}(i,j)}_j$ ($j\in[1,s]$)             & $\le2\alpha$      & $\left(b_i\to b_{\text{adj}(i,j)}\to a^{\text{bin}(i,j)}_j\right)$ \\
											& $a^\star_j$ ($j\in[1,\ell]$)                      & $\le\beta$        & $\left(b_i\to a^\star_j\right)$ \\
\midrule
router 										& router                                            & $\le2\alpha$      & $\left(u\to t\to v\right)$ \\
\bottomrule[1.5pt]
\end{tabular}
\end{table*}

Regarding the distance between $a_i$ and $b_i$ for $i\in[1,2^s]$, if there exists $j\in[1,\ell]$ such that $x_{i,j}=y_{i,j}=1$, then $w(\{a_i,a^\star_j\})=w(\{b_i,a^\star_j\})=\alpha$ and $d_{G',w}(a_i,b_i)\le2\alpha$ because of the path $(a_i\to a^\star_j\to b_i)$ in $G'$.
If there is no $j\in[1,\ell]$ such that $x_{i,j}=y_{i,j}=1$, we claim that  $d_{G',w}(a_i,b_i)\geq\min\{\alpha+\beta, 3\alpha\}$.
For any path between $a_i$ and $b_i$, if it contains exactly two edges, it is of the form $(a_i\to a^\star_j\to b_i)$ for some $j\in[1,\ell]$ by the construction of $G'$, and it is of length at least $\alpha+\beta$ by the assumption.
If it contains at least three edges, it is of length at least $3\alpha$.

If $F(x,y)=1$, then for any $i\in[1,2^s]$, there exists $j\in[1,\ell]$ such that $x_{i,j}=y_{i,j}=1$.
Hence,
$$
\begin{aligned}
& d_{G',w}(a_i,b_i)\le2\alpha,\forall i\in[1,2^s], \\
& D_{G',w}=\max_{u,v}d_{G',w}(u,v)\le\max\{2\alpha,\beta\}.
\end{aligned}
$$
Therefore, $D_{G,w}\le D_{G',w}+n\le\max\{2\alpha,\beta\}+n$ by Lemma~\ref{lem:contraction}.

If $F(x,y)=0$, then there exists $i\in[1,2^s]$ such that $x_{i,j}=0$ or $y_{i,j}=0$ for any $j\in[1,\ell]$.
Hence,
$$
\begin{aligned}
& d_{G',w}(a_i,b_i)=\min_{\text{path }P\text{ from }a_i\text{ to }b_i}\text{length}(P)\ge\min\{\alpha+\beta,3\alpha\}, \\
& D_{G',w}=\max_{u,v}d_{G',w}(u,v)\ge d_{G',w}(a_i,b_i)\ge\min\{\alpha+\beta,3\alpha\}.
\end{aligned}
$$
Therefore, $D_{G,w}\ge D_{G',w}\ge\min\{\alpha+\beta,3\alpha\}$ by Lemma~\ref{lem:contraction}.
\end{proof}

Combining Lemma~\ref{lem:simulation} and Lemma~\ref{lem:reduction}, we have a reduction from computing $F$ in the Server model to approximating diameter in the quantum CONGEST model.
To prove the communication complexity of $F$ in the Server model, we adopt the following lemma.

\begin{lemma}[Lemma B.4 in \cite{ElkinKNP14}, arXiv version]
Function ${\rm VER}:\{0,1,2,3\}\times\{0,1,2,3\}\to\{0,1\}$ is defined by ${\rm VER}(x,y)=1$ if and only if $x+y$ is equivalent to $0$ or $1$ modulo $4$, where $x,y\in\{0,1,2,3\}$.
Let $f:\{0,1\}^k\to\{0,1\}$ be an arbitrary function.
Then
$$Q^{sv}_\varepsilon(f\circ{\rm VER}^k)\ge\frac12\text{deg}_{4\varepsilon}(f)-O(1)$$
for any $0<\varepsilon<1/4$.
\label{lem:lifting}
\end{lemma}

A read-once formula, which consists of AND gates, OR gates, and NOT gates, is a formula in which each variable appears exactly once.
We will need the following conclusion for approximate degree of read-once formulas.

\begin{lemma}[Theorem 6 in \cite{AaronsonBKRT21}]
For any read-once formula $f:\{0,1\}^k\to\{0,1\}$, $\text{deg}_{1/3}(f)=\Theta\left(\sqrt k\right)$.
\label{lem:read_once}
\end{lemma}

\begin{lemma}
Given $s,\ell$ defined in Eq.~\eqref{eqn:paramters} where $\ell$ is a multiple of $4$, $F={\rm AND}_{2^s}\circ({\rm OR}_\ell\circ{\rm AND}^\ell_2)^{2^s}$ with inputs $x,y\in\{0,1\}^{2^s\cdot\ell}$, set
\[F(x,y)=\bigwedge_{i\in[1,2^s]}\left(\bigvee_{j\in[1,\ell]}\left(x_{i,j}\wedge y_{i,j}\right)\right).\]
It holds that
$$Q^{sv}_{1/12}(F)=\Omega\left(\sqrt{2^s\cdot\ell}\right).$$
\label{lem:and_or_and}
\end{lemma}
\begin{proof}
The function $F$ can be rewritten as $F=f\circ \text{GDT}^{2^s\cdot\ell/4}$, where
$f=\text{AND}_{2^s}\circ\text{OR}^{2^s}_{\ell/4}$ and $\text{GDT}=\text{OR}_4\circ\text{AND}^4_2$.
Obviously the function $f$ is a read-once formula.
It can be seen that the function VER is actually a {\it promise version} of the function GDT where inputs $x,y\in\{0,1\}^4$ satisfy
$$x\in\{0011,1001,1100,0110\},y\in\{0001,0010,0100,1000\}.$$
Thus, the lower bound for $f\circ\text{VER}^{2^s\cdot\ell/4}$ clearly implies the lower bound for $f\circ\text{GDT}^{2^s\cdot\ell/4}$.
Therefore,
$$Q^{sv}_{1/12}(f\circ\text{GDT}^{2^s\cdot\ell/4})\ge Q^{sv}_{1/12}(f\circ\text{VER}^{2^s\cdot\ell/4})\ge\frac12\text{deg}_{1/3}(f)-O(1)=\Omega\left(\sqrt{2^s\cdot\ell}\right).$$
The second inequality is due to Lemma~\ref{lem:lifting} and the last inequality is due to Lemma~\ref{lem:read_once}.
\end{proof}

\begin{proof}[Proof of Theorem~\ref{thm:diameter}]
Let $\mathcal A$ be a $T$-round algorithm ($T<2^h/2$) in the quantum CONGEST model which, for any weighted graph $(G,w)$, computes a $(\frac32-\varepsilon)$-approximation of $D_{G,w}$ (constant $\varepsilon\in(0,1/2]$) with probability at least $11/12$.
Alice and Bob, who receive $x,y\in\{0,1\}^{2^s\cdot\ell}$, respectively, construct the network $G$ as described above with parameters $h,s,\ell$ given in Eq.~\eqref{eqn:paramters}.
The number of nodes is
$$n=\left(2^{h+1}-1\right)+\left(2s+\ell\right)\left(2^h+2\right)+2\cdot2^s=\Theta\left(2^{3h/2}\right).$$
And the unweighted diameter is $D_G=\Theta(h)=\Theta(\log n)$.
Let $w$ be the weight function.
Due to Lemma~\ref{lem:simulation}, they can simulate $\mathcal A$ on $(G,w)$ in the quantum Server model with communication complexity $O(T\cdot h\cdot B)$ where $B$ denotes the bandwidth.
With probability at least $\frac{11}{12}$, Alice and Bob output an approximation $\widetilde D_{G,w}$ satisfying $D_{G,w}\le\widetilde D_{G,w}\le(\frac32-\varepsilon)D_{G,w}$.
We set $\alpha=n^2$ and $\beta=2n^2$.
By Lemma~\ref{lem:reduction},
$$
\begin{aligned}
\text{if }F(x,y)=1,\widetilde D_{G,w} & \le\left(\frac32-\varepsilon\right)D_{G,w}\le\left(\frac32-\varepsilon\right)\left(\max\{2\alpha,\beta\}+n\right) \\
& =3n^2-\left(2\varepsilon n^2-\left(\frac32-\varepsilon\right)n\right); \\ %\text{ and} \\
\text{if }F(x,y)=0,\widetilde D_{G,w} & \ge D_{G,w}\ge\min\{\alpha+\beta,3\alpha\}=3n^2. \\
\end{aligned}
$$
For large enough $n$, Alice and Bob can distinguish whether $F(x,y)=1$ or not with probability at least $\frac{11}{12}$ in the Server model, and thus $Q^{sv}_{1/12}(F)=O(T\cdot h\cdot B)$.
Due to Lemma~\ref{lem:and_or_and},
$$T=\Omega\left(\frac{\sqrt{2^s\cdot\ell}}{h\cdot B}\right)=\Omega\left(\frac{2^h}{h\cdot B}\right)=\Omega\left(\frac{n^{2/3}}{ \log^2 n}\right),$$
where the last equality is by the choice of $h$ and the the bandwidth $B=\Theta(\log n)$.
Therefore, the round complexity of approximating diameter is $\Omega\left(\min\left\{2^h/2,\frac{n^{2/3}}{\log^2n}\right\}\right)=\Omega\left(\frac{n^{2/3}}{\log^2n}\right)$.
\end{proof}

\subsection{Hardness of Approximating Radius}

We choose the same set of parameters $h,s,\ell$ given in Eq.~\eqref{eqn:paramters}.
The argument is very close to the one for diameter.

\begin{theorem}[Restated]
For any constant $\varepsilon\in(0,\frac12]$, any algorithm, with probability at least $\frac{11}{12}$, computing a $(\frac32-\varepsilon)$-approximation of radius in the quantum CONGEST model requires $\Omega\left(\frac{n^{2/3}}{\log^2n}\right)$ rounds, even when the unweighted diameter is $\Theta(\log n)$, where $n$ denotes the number of nodes.
\label{thm:radius}
\end{theorem}

The weighted graph $(G,w)$ that we construct for showing hardness of approximating radius is almost the same except that we add a node $a_0$ in $V_A$ along with edges $\{a_0,a_1\},\cdots,\{a_0,a_{2^s}\}$ of weight $2\alpha$.
Here we only show in Figure~\ref{fig:radius} the graph $G'$ after contracting all edges of weight $1$ (the green edges are the new-added edges).

\begin{figure}[ht]
\centering
\includegraphics[width=0.6\linewidth]{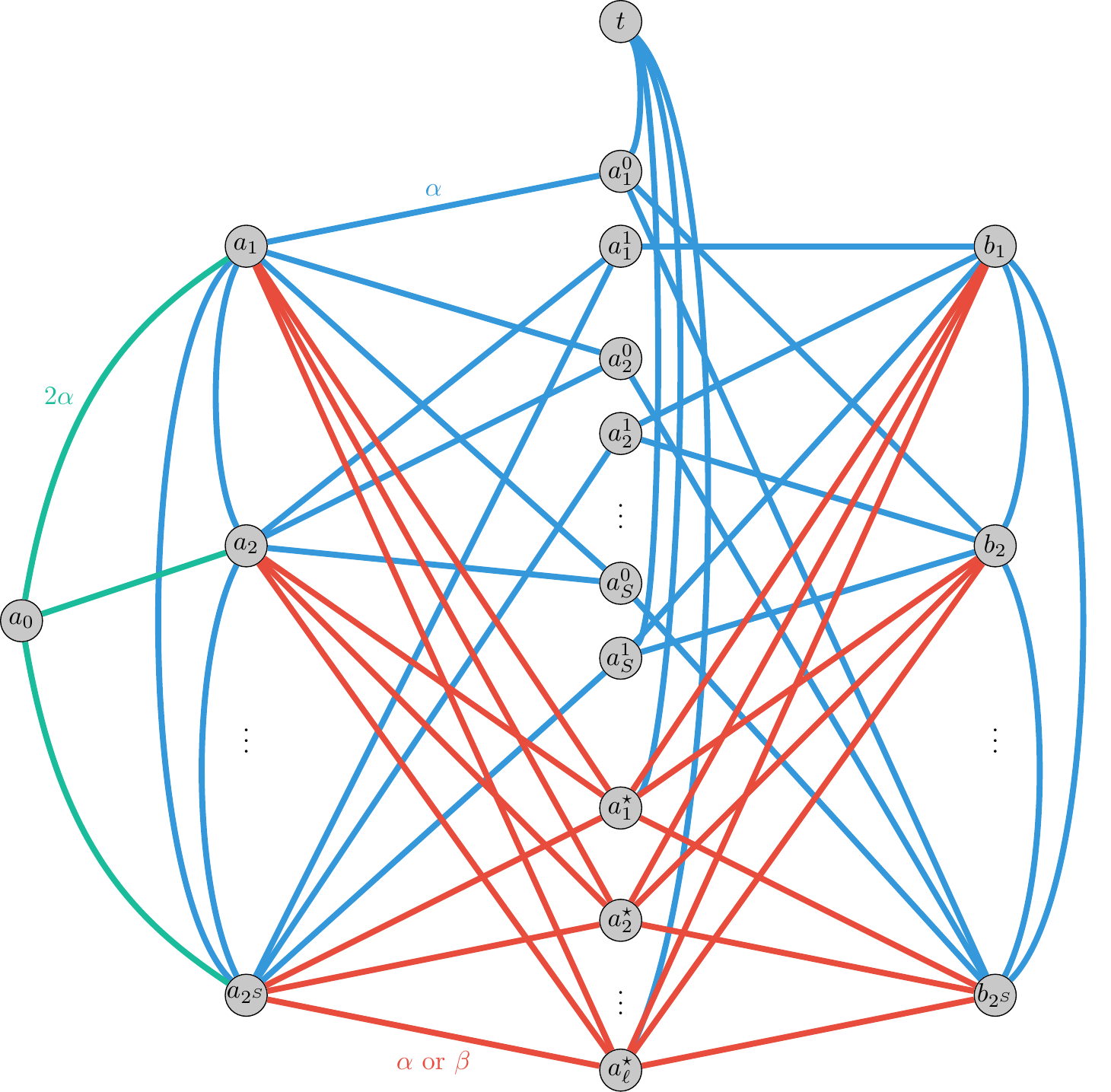}
\caption{
Graph $G'$ (after contraction) for approximating radius.
The additional green edges are of weight $2\alpha$.
The eccentricity of any node, except $a_i$ for $i\in[1,2^s]$, is at least $3\alpha$; and the eccentricity of $a_i$ is at most $\max\{2\alpha,\beta\}$ if there exists $j\in[1,\ell]$ such that $x_{i,j}=y_{i,j}=1$, otherwise it is at least $\min\{\alpha+\beta,3\alpha\}$.
Therefore, the radius is at most $\max\{2\alpha,\beta\}$ if there exist $i\in[1,2^s],j\in[1,\ell]$ such that $x_{i,j}=y_{i,j}=1$, otherwise it is at least $\min\{\alpha+\beta,3\alpha\}$.
}
\label{fig:radius}
\end{figure}

For inputs $x,y\in\{0,1\}^{2^s\cdot\ell}$ define
$$F'(x,y)=\bigvee_{i\in[1,2^s],j\in[1,\ell]}(x_{i,j}\wedge y_{i,j}).$$
We have the following lemma.

\begin{lemma}
$R_{G,w}\le\max\{2\alpha,\beta\}+n$ if $F'(x,y)=1$, and $R_{G,w}\ge\min\{\alpha+\beta,3\alpha\}$ otherwise.
\label{lem:reduction_radius}
\end{lemma}
\begin{proof}
It suffices to estimate the radius of $(G',w)$ by Lemma~\ref{lem:contraction}.
For any node $v\notin\{a_0,a_1,\cdots,a_{2^s}\}$, $d_{G',w}(a_0,v)\ge3\alpha$.
This is because that any path from $a_0$ to $v$ is of the form $(a_0\to a_i\leadsto v)$ for some $i\in[1,2^s]$, where $w(\{a_0,a_i\})=2\alpha$, and the remaining edges on the path have total weight at least $\alpha$.
Therefore, $e_{G',w}(v)\ge3\alpha$ for any $v\notin\{a_1,\cdots,a_{2^s}\}$.
To estimate the eccentricity of $a_i$ for $i\in[1,2^s]$, we have $d(a_i,v)\le\max\{2\alpha,\beta\}$ for any $v\ne b_i$ as shown on Table~\ref{tab:distance}, and $d_{G',w}(a_i,b_i)\le2\alpha$ if there exists $j\in[1,\ell]$ such that $x_{i,j}=y_{i,j}=1$, and $d_{G',w}(a_i,b_i)\ge\min\{\alpha+\beta,3\alpha\}$ otherwise.

If $F'(x,y)=1$, then there are $i\in[1,2^s]$ and $j\in[1,\ell]$ such that $x_{i,j}=y_{i,j}=1$, and thus
$$
\begin{aligned}
& d_{G',w}(a_i,b_i)\le2\alpha,\\
& e_{G',w}(a_i)=\max_v d_{G',w}(a_i,v)\le\max\{2\alpha,\beta\}, \\
& R_{G',w}=\min_u e_{G',w}(u)\le e_{G',w}(a_i)\le\max\{2\alpha,\beta\}.
\end{aligned}
$$
Therefore, $R_{G,w}\le R_{G',w}+n\le\max\{2\alpha,\beta\}+n$ by Lemma~\ref{lem:contraction}.

If $F'(x,y)=0$, then for any $i\in[1,2^s]$ and $j\in[1,\ell]$, $x_{i,j}=0$ or $y_{i,j}=0$, and thus
$$
\begin{aligned}
& d_{G',w}(a_i,b_i)\ge\min\{\alpha+\beta,3\alpha\},\forall i\in[1,2^s], \\
& e_{G',w}(a_i)=\max_v d_{G',w}(a_i,v)\ge d_{G',w}(a_i,b_i) \\
& \qquad\ge\min\{\alpha+\beta,3\alpha\},\forall i\in[1,2^s], \\
& R_{G',w}=\min_u e_{G',w}(u)\ge\min\{\alpha+\beta,3\alpha\}.
\end{aligned}
$$
Therefore, $R_{G,w}\ge R_{G',w}\ge\min\{\alpha+\beta,3\alpha\}$ by Lemma~\ref{lem:contraction}.
\end{proof}

Similar to Lemma~\ref{lem:and_or_and}, one can prove a lower bound on communication complexity of $F'$ in the quantum Server model.

\begin{lemma}
Given $s,\ell$ defined in Eq.~\eqref{eqn:paramters} where $2^s\cdot\ell$ is a multiple of $4$,
$F'={\rm OR}_{2^s\cdot\ell}\circ{\rm AND}^{2^s\cdot\ell}_2$ with inputs $x,y\in\{0,1\}^{2^s\cdot\ell}$,
set
$$F'(x,y)=\bigvee_{i\in[1,2^s],j\in[1,\ell]}(x_{i,j}\wedge y_{i,j}).$$
It holds that
$$Q^{sv}_{1/12}(F')=\Omega\left(\sqrt{2^s\cdot\ell}\right).$$
\label{lem:or_and}
\end{lemma}
\begin{proof}
The function $F'$ can be rewritten as $F'=f'\circ \text{GDT}^{2^s\cdot\ell/4}$,
where $f'=\text{OR}_{2^s\cdot\ell/4}$.
Note that $f'$ is still a read-once formula.
Thus the rest of proof is the same as the one in Lemma~\ref{lem:and_or_and}.
\end{proof}

\begin{proof}[Proof of Theorem~\ref{thm:radius}]
Let $\mathcal A$ be a $T$-round algorithm ($T<2^h/2$) in the quantum CONGEST model which, for any weighted graph $(G,w)$, computes a $(\frac32-\varepsilon)$-approximation of $R_{G,w}$ (constant $\varepsilon\in(0,\frac12]$) with probability at least $\frac{11}{12}$.
Alice and Bob, who receive $x,y\in\{0,1\}^{2^s\cdot\ell}$ as input, construct the weighted graph $(G,w)$ described above with the number of node $n=\Theta(2^{3h/2})$. The unweighted diameter $D_G=\Theta(\log n)$.
Due to Lemma~\ref{lem:simulation}, Alice and Bob can simulate $\mathcal A$ on $(G,w)$ in the quantum Server model with communication complexity $O(T\cdot h\cdot B)$.
Then with probability at least $\frac{11}{12}$, Alice and Bob compute $\widetilde R_{G,w}$ satisfying $R_{G,w}\le\widetilde R_{G,w}\le(\frac32-\varepsilon)R_{G,w}$.
We set $\alpha=n^2$ and $\beta=2n^2$.
Due to Lemma~\ref{lem:reduction_radius},
$$
\begin{aligned}
& \text{if }F'(x,y)=1,\widetilde R_{G,w}\le3n^2-\left(2\varepsilon n^2-\left(\frac32-\varepsilon\right)n\right); \\
& \text{if }F'(x,y)=0,\widetilde R_{G,w}\ge3n^2. \\
\end{aligned}
$$
For large enough $n$, Alice and Bob can compute $F'$ with probability at least $\frac{11}{12}$ in the Server model, and thus $Q^{sv}_{1/12}(F')=O(T\cdot h\cdot B)$.
Due to Lemma~\ref{lem:or_and}, $T=\Omega\left(\frac{n^{2/3}}{\log^2n}\right)$.
Therefore, the round complexity of approximating radius is $\Omega\left(\min\left\{2^h/2,\frac{n^{2/3}}{\log^2n}\right\}\right)=\Omega\left(\frac{n^{2/3}}{\log^2n}\right)$.
\end{proof}

\section*{Acknowledgements}
We thank the anonymous reviewers' feedback.
This work was supported in part by, the National Key R\&D Program of China 2018YFB1003202, National Natural Science Foundation of China (Grant No. 61972191), the Program for Innovative Talents and Entrepreneur in Jiangsu, and Anhui Initiative in Quantum Information Technologies (Grant No. AHY150100).

\bibliographystyle{acm}
\bibliography{main}

\begin{thebibliography}{10}

\bibitem{AaronsonBKRT21}
{\sc Aaronson, S., Ben{-}David, S., Kothari, R., Rao, S., and Tal, A.}
\newblock Degree vs. approximate degree and quantum implications of huang's
  sensitivity theorem.
\newblock In {\em {STOC} '21: 53rd Annual {ACM} {SIGACT} Symposium on Theory of
  Computing, Virtual Event, Italy, June 21-25, 2021\/} (2021), S.~Khuller and
  V.~V. Williams, Eds., {ACM}, pp.~1330--1342.

\bibitem{AbboudCK16}
{\sc Abboud, A., Censor{-}Hillel, K., and Khoury, S.}
\newblock Near-linear lower bounds for distributed distance computations, even
  in sparse networks.
\newblock In {\em Distributed Computing - 30th International Symposium, {DISC}
  2016, Paris, France, September 27-29, 2016. Proceedings\/} (2016),
  C.~Gavoille and D.~Ilcinkas, Eds., vol.~9888 of {\em Lecture Notes in
  Computer Science}, Springer, pp.~29--42.

\bibitem{AnconaCDEW20}
{\sc Ancona, B., Censor{-}Hillel, K., Dalirrooyfard, M., Efron, Y., and
  Williams, V.~V.}
\newblock Distributed distance approximation.
\newblock In {\em 24th International Conference on Principles of Distributed
  Systems, {OPODIS} 2020, December 14-16, 2020, Strasbourg, France (Virtual
  Conference)\/} (2020), Q.~Bramas, R.~Oshman, and P.~Romano, Eds., vol.~184 of
  {\em LIPIcs}, Schloss Dagstuhl - Leibniz-Zentrum f{\"{u}}r Informatik,
  pp.~30:1--30:17.

\bibitem{Ben-OrH05}
{\sc Ben{-}Or, M., and Hassidim, A.}
\newblock Fast quantum byzantine agreement.
\newblock In {\em Proceedings of the 37th Annual {ACM} Symposium on Theory of
  Computing, Baltimore, MD, USA, May 22-24, 2005\/} (2005), H.~N. Gabow and
  R.~Fagin, Eds., {ACM}, pp.~481--485.

\bibitem{Bennett89}
{\sc Bennett, C.~H.}
\newblock Time/space trade-offs for reversible computation.
\newblock {\em {SIAM} J. Comput. 18}, 4 (1989), 766--776.

\bibitem{BernsteinN19}
{\sc Bernstein, A., and Nanongkai, D.}
\newblock Distributed exact weighted all-pairs shortest paths in near-linear
  time.
\newblock In {\em Proceedings of the 51st Annual {ACM} {SIGACT} Symposium on
  Theory of Computing, {STOC} 2019, Phoenix, AZ, USA, June 23-26, 2019\/}
  (2019), M.~Charikar and E.~Cohen, Eds., {ACM}, pp.~334--342.

\bibitem{Censor-HillelFG22}
{\sc Censor{-}Hillel, K., Fischer, O., Gall, F.~L., Leitersdorf, D., and
  Oshman, R.}
\newblock Quantum distributed algorithms for detection of cliques.
\newblock In {\em 13th Innovations in Theoretical Computer Science Conference,
  {ITCS} 2022, January 31 - February 3, 2022, Berkeley, CA, {USA}\/} (2022),
  M.~Braverman, Ed., vol.~215 of {\em LIPIcs}, Schloss Dagstuhl -
  Leibniz-Zentrum f{\"{u}}r Informatik, pp.~35:1--35:25.

\bibitem{ChechikM20}
{\sc Chechik, S., and Mukhtar, D.}
\newblock Single-source shortest paths in the {CONGEST} model with improved
  bound.
\newblock In {\em {PODC} '20: {ACM} Symposium on Principles of Distributed
  Computing, Virtual Event, Italy, August 3-7, 2020\/} (2020), Y.~Emek and
  C.~Cachin, Eds., {ACM}, pp.~464--473.

\bibitem{Elkin06}
{\sc Elkin, M.}
\newblock An unconditional lower bound on the time-approximation trade-off for
  the distributed minimum spanning tree problem.
\newblock {\em {SIAM} J. Comput. 36}, 2 (2006), 433--456.

\bibitem{ElkinKNP14}
{\sc Elkin, M., Klauck, H., Nanongkai, D., and Pandurangan, G.}
\newblock Can quantum communication speed up distributed computation?
\newblock In {\em {ACM} Symposium on Principles of Distributed Computing,
  {PODC} '14, Paris, France, July 15-18, 2014\/} (2014), M.~M.
  Halld{\'{o}}rsson and S.~Dolev, Eds., {ACM}, pp.~166--175.

\bibitem{FrischknechtHW12}
{\sc Frischknecht, S., Holzer, S., and Wattenhofer, R.}
\newblock Networks cannot compute their diameter in sublinear time.
\newblock In {\em Proceedings of the Twenty-Third Annual {ACM-SIAM} Symposium
  on Discrete Algorithms, {SODA} 2012, Kyoto, Japan, January 17-19, 2012\/}
  (2012), Y.~Rabani, Ed., {SIAM}, pp.~1150--1162.

\bibitem{GallM18}
{\sc Gall, F.~L., and Magniez, F.}
\newblock Sublinear-time quantum computation of the diameter in {CONGEST}
  networks.
\newblock In {\em Proceedings of the 2018 {ACM} Symposium on Principles of
  Distributed Computing, {PODC} 2018, Egham, United Kingdom, July 23-27,
  2018\/} (2018), C.~Newport and I.~Keidar, Eds., {ACM}, pp.~337--346.

\bibitem{GallNR19}
{\sc Gall, F.~L., Nishimura, H., and Rosmanis, A.}
\newblock Quantum advantage for the {LOCAL} model in distributed computing.
\newblock In {\em 36th International Symposium on Theoretical Aspects of
  Computer Science, {STACS} 2019, March 13-16, 2019, Berlin, Germany\/} (2019),
  R.~Niedermeier and C.~Paul, Eds., vol.~126 of {\em LIPIcs}, Schloss Dagstuhl
  - Leibniz-Zentrum f{\"{u}}r Informatik, pp.~49:1--49:14.

\bibitem{GavoilleKM09}
{\sc Gavoille, C., Kosowski, A., and Markiewicz, M.}
\newblock What can be observed locally?
\newblock In {\em Distributed Computing, 23rd International Symposium, {DISC}
  2009, Elche, Spain, September 23-25, 2009. Proceedings\/} (2009), I.~Keidar,
  Ed., vol.~5805 of {\em Lecture Notes in Computer Science}, Springer,
  pp.~243--257.

\bibitem{HolzerPRW14}
{\sc Holzer, S., Peleg, D., Roditty, L., and Wattenhofer, R.}
\newblock Distributed 3/2-approximation of the diameter.
\newblock In {\em Distributed Computing - 28th International Symposium, {DISC}
  2014, Austin, TX, USA, October 12-15, 2014. Proceedings\/} (2014), F.~Kuhn,
  Ed., vol.~8784 of {\em Lecture Notes in Computer Science}, Springer,
  pp.~562--564.

\bibitem{HolzerP15}
{\sc Holzer, S., and Pinsker, N.}
\newblock Approximation of distances and shortest paths in the broadcast
  congest clique.
\newblock In {\em 19th International Conference on Principles of Distributed
  Systems, {OPODIS} 2015, December 14-17, 2015, Rennes, France\/} (2015),
  E.~Anceaume, C.~Cachin, and M.~G. Potop{-}Butucaru, Eds., vol.~46 of {\em
  LIPIcs}, Schloss Dagstuhl - Leibniz-Zentrum f{\"{u}}r Informatik,
  pp.~6:1--6:16.

\bibitem{HolzerW12}
{\sc Holzer, S., and Wattenhofer, R.}
\newblock Optimal distributed all pairs shortest paths and applications.
\newblock In {\em {ACM} Symposium on Principles of Distributed Computing,
  {PODC} '12, Funchal, Madeira, Portugal, July 16-18, 2012\/} (2012),
  D.~Kowalski and A.~Panconesi, Eds., {ACM}, pp.~355--364.

\bibitem{IzumiG19}
{\sc Izumi, T., and Gall, F.~L.}
\newblock Quantum distributed algorithm for the all-pairs shortest path problem
  in the {CONGEST-CLIQUE} model.
\newblock In {\em Proceedings of the 2019 {ACM} Symposium on Principles of
  Distributed Computing, {PODC} 2019, Toronto, ON, Canada, July 29 - August 2,
  2019\/} (2019), P.~Robinson and F.~Ellen, Eds., {ACM}, pp.~84--93.

\bibitem{IzumiGM20}
{\sc Izumi, T., Gall, F.~L., and Magniez, F.}
\newblock Quantum distributed algorithm for triangle finding in the {CONGEST}
  model.
\newblock In {\em 37th International Symposium on Theoretical Aspects of
  Computer Science, {STACS} 2020, March 10-13, 2020, Montpellier, France\/}
  (2020), C.~Paul and M.~Bl{\"{a}}ser, Eds., vol.~154 of {\em LIPIcs}, Schloss
  Dagstuhl - Leibniz-Zentrum f{\"{u}}r Informatik, pp.~23:1--23:13.

\bibitem{MagniezN20}
{\sc Magniez, F., and Nayak, A.}
\newblock Quantum distributed complexity of set disjointness on a line.
\newblock In {\em 47th International Colloquium on Automata, Languages, and
  Programming, {ICALP} 2020, July 8-11, 2020, Saarbr{\"{u}}cken, Germany
  (Virtual Conference)\/} (2020), A.~Czumaj, A.~Dawar, and E.~Merelli, Eds.,
  vol.~168 of {\em LIPIcs}, Schloss Dagstuhl - Leibniz-Zentrum f{\"{u}}r
  Informatik, pp.~82:1--82:18.

\bibitem{Nanongkai14STOC}
{\sc Nanongkai, D.}
\newblock Distributed approximation algorithms for weighted shortest paths.
\newblock In {\em Symposium on Theory of Computing, {STOC} 2014, New York, NY,
  USA, May 31 - June 03, 2014\/} (2014), D.~B. Shmoys, Ed., {ACM},
  pp.~565--573.

\bibitem{PelegRT12}
{\sc Peleg, D., Roditty, L., and Tal, E.}
\newblock Distributed algorithms for network diameter and girth.
\newblock In {\em Automata, Languages, and Programming - 39th International
  Colloquium, {ICALP} 2012, Warwick, UK, July 9-13, 2012, Proceedings, Part
  {II}\/} (2012), A.~Czumaj, K.~Mehlhorn, A.~M. Pitts, and R.~Wattenhofer,
  Eds., vol.~7392 of {\em Lecture Notes in Computer Science}, Springer,
  pp.~660--672.

\bibitem{SarmaHKKNPPW12}
{\sc Sarma, A.~D., Holzer, S., Kor, L., Korman, A., Nanongkai, D., Pandurangan,
  G., Peleg, D., and Wattenhofer, R.}
\newblock Distributed verification and hardness of distributed approximation.
\newblock {\em {SIAM} J. Comput. 41}, 5 (2012), 1235--1265.

\bibitem{TaniKM12}
{\sc Tani, S., Kobayashi, H., and Matsumoto, K.}
\newblock Exact quantum algorithms for the leader election problem.
\newblock {\em {ACM} Trans. Comput. Theory 4}, 1 (2012), 1:1--1:24.

\end{thebibliography}

\appendix

\section{Toolkits in Nanongkai's Algorithm}
\label{sec:toolkits}

Let $G=(V,E)$ be a distributed network with a weight function $w:E\to\mathbb N^+$ and a pre-defined node $\text{leader}\in V$.
We assume that each node initially knows $n=|V|$ and $W=\max_{e\in E}w(e)$.
The parameters $\varepsilon,r,\ell,k$ are chosen the same as in Eq.~\eqref{eqn:alg_parameters}.
We follow the background of Section~\ref{sec:approx_ecc}.
Given a vertex set $S\subseteq V$, let $\widetilde d^\ell_{G,w}(\cdot)$, $(G'_S,w'_S)$, $N^k_S(\cdot)$, $(G''_S,w''_S)$, $\widetilde d_{G,w,S}(\cdot)$ be as defined in Lemma~\ref{lem:bounded_hop_distance} and Lemma~\ref{lem:approx_distance}.
The following lemmas and algorithms are summarized from \cite[arXiv version]{Nanongkai14STOC}.

\begin{lemma}[Theorem 3.2 in \cite{Nanongkai14STOC}]
For $s\in V$ known to all nodes, there exists an algorithm (Algorithm~\ref{alg:bounded_hop_sssp}) such that in $\widetilde O(\ell/\varepsilon)$ rounds, each $v\in V$ knows $\widetilde d^\ell_{G,w}(s,v)$, and during the whole computation, each node broadcasts $O(\log n)$ messages of size $O(\log n)$ to its neighbors.
\label{lem:bounded_hop_sssp}
\end{lemma}

\begin{algorithm}
\caption{Bounded-Hop SSSP $(G,w,s,\ell,\varepsilon)$}
\label{alg:bounded_hop_sssp}
\begin{algorithmic}[1]
\Require Network $(G,w)$, source node $s$ and parameters $\ell,\varepsilon>0$.
\Ensure Each node $v$ knows $\widetilde d^\ell_{G,w}(s,v)$.
\State Initially, $\widetilde d^\ell_{G,w}(s,v)\leftarrow\infty$ for each $v\in V$.
\For{$i=0$ to $\log\frac{2nW}\varepsilon$}
\State Run bounded-distance SSSP with parameters $(G,w_i,s,(1+2/\varepsilon)\ell)$ using Algorithm~\ref{alg:bounded_distance_sssp}.
\For{each $v\in V$} in parallel
\If{$d_{G,w_i}(s,v)\le(1+2/\varepsilon)\ell$}
\State $\widetilde d^\ell_{G,w}(s,v)\leftarrow\min\left\{\widetilde d^\ell_{G,w}(s,v),d_{G,w_i}(s,v)\right\}$.
\EndIf
\EndFor
\EndFor
\end{algorithmic}
\end{algorithm}

\begin{algorithm}
\caption{Bounded-Distance SSSP $(G,w,s,L)$}
\label{alg:bounded_distance_sssp}
\begin{algorithmic}[1]
\Require Network $(G,w)$, source node $s$ and parameter $L>0$.
\Ensure Each node $v$ knows whether $d_{G,w}(s,v)\le L$, and if so, it further knows $d_{G,w}(s,v)$.
\State Initially, $d_{G,w}(s,s)\leftarrow0$ and $d_{G,w}(s,v)\leftarrow\infty$ for each $v\ne s$.
\State Let $t$ be the time this algorithm starts.
\For{round $r=t$ to $t+L$}
\For{each $v\in V$} in parallel
\For{each message $(u,d_{G,w}(s,u))$ received in the previous round}
\If{$d_{G,w}(s,u)+w(\{u,v\})\le L$}
\State $d_{G,w}(s,v)\leftarrow\min\left\{d_{G,w}(s,v),d_{G,w}(s,u)+w(\{u,v\})\right\}$.
\EndIf
\EndFor
\If{$d_{G,w}(s,v)=r-t$}
\State $v$ broadcasts message $(v,d_{G,w}(s,v))$ to all neighbors.
\EndIf
\EndFor
\EndFor
\end{algorithmic}
\end{algorithm}

\begin{lemma}[Theorem 3.6 and Lemma 3.7 in \cite{Nanongkai14STOC}]
There exist an algorithm (Algorithm~\ref{alg:bounded_hop_mssp}) such that in $\widetilde O(D_G+\ell/\varepsilon+|S|)$ rounds, each node $v\in V$ knows $\widetilde d^\ell_{G,w}(s,v)$ for each $s\in S$, with probability of failure at most $n^{-c}$, for any constant $c>0$ and sufficiently large $n$.
\label{lem:bounded_hop_mssp}
\end{lemma}

\begin{algorithm}[!h]
\caption{Bounded-Hop Multi-Source Shortest Paths $(G,w,S,\ell,\varepsilon)$}
\label{alg:bounded_hop_mssp}
\begin{algorithmic}[1]
\Require Network $(G,w)$, set of source nodes $S$ and parameters $\ell,\varepsilon>0$.
\Ensure With high probability, each node $v$ knows $\widetilde d^\ell_{G,w}(s,v)$ for each $s\in S$.
\State Assume that $S=\{s_1,\cdots,s_b\}$.
Let $\mathcal A_i$ be the Algorithm~\ref{alg:bounded_hop_sssp} with parameters $(G,w,s_i,\ell,\varepsilon)$ for each $i\in[1,k]$ (each $\mathcal A_i$ is of $T=\widetilde O(\ell/\varepsilon)$ rounds, and during the whole computation of $\mathcal A_i$, each node broadcasts $O(\log n)$ messages to its neighbors due to Lemma~\ref{lem:bounded_hop_sssp}).
\State The node leader samples $\Delta_1,\cdots,\Delta_b\in[0,b\log n]$ independently and uniformly at random for delaying algorithms $\mathcal A_1,\cdots,\mathcal A_k$, and broadcasts them by pipelining in $O(D_G+b)$ rounds.
\For{$r=1$ to $T+b\log n$}
\For{each $v\in V$} in parallel
\State Let $a=|\ \{i\in[1,b]:v\text{ broadcasts a message in the }(r-\Delta_i)\text{-th round of }\mathcal A_i\}\ |$.
\If{$a\le\lceil\log n\rceil$}
\State $v$ broadcasts these $a$ messages in the next $\lceil\log n\rceil$ rounds.
\Else
\State The algorithm fails.
\EndIf
\EndFor
\EndFor
\end{algorithmic}
\end{algorithm}

\newpage

\begin{lemma}[Theorem 4.5 in \cite{Nanongkai14STOC}]
After the overlay network $(G'_S,w'_S)$ is embedded, there exists an algorithm (Algorithm~\ref{alg:embedding}) which further embeds the overlay network $(G''_S,w''_S)$ in $\widetilde O(D_G+|S|k)$
rounds.
\label{lem:embedding}
\end{lemma}

\begin{algorithm}
\caption{Embedding Overlay Network $(G,w,S,G'_S,w'_S,k)$}
\label{alg:embedding}
\begin{algorithmic}[1]
\Require Network $(G,w)$, set of source nodes $S$, overlay network $(G'_S,w'_S)$ and parameter $k>0$.
\Ensure It embeds the overlay network $(G''_S,w''_S)$.
\State Each node $s\in S$ broadcasts the $k$ shortest edges incident to it on $(G'_S,w'_S)$ (this can be done in $O(D_G+|S|k)$ rounds).
\For{each $s\in S$} locally
\State $s$ computes $N^k_S(s)$, along with the weight $w''_S(\{s,v\})=d_{G'_S,w'_S}(s,v)$ for each $v\in N^k_S(s)$ (this can be done due to Observation 3.12 in \cite{Nanongkai14STOC}).
\EndFor
\end{algorithmic}
\end{algorithm}

\begin{lemma}[Lemma 4.6 in \cite{Nanongkai14STOC}]
For node $s\in S$ known to all nodes, after the overlay network $(G''_S,w''_S)$ is embedded, there exists an algorithm (Algorithm~\ref{alg:sssp_on_overlay}) such that in $\widetilde O\left(\frac{|S|}{\varepsilon\cdot k}\cdot D_G+|S|\right)$ rounds, each node $v\in V$ knows for each $u\in S$ the value of $\widetilde d^{4|S|/k}_{G''_S,w''_S}(s,u)$.
\label{lem:sssp_on_overlay}
\end{lemma}

\begin{algorithm}
\caption{SSSP on Overlay Network $(G,w,S,\varepsilon,k,G''_S,w''_S,s)$}
\label{alg:sssp_on_overlay}
\begin{algorithmic}[1]
\Require Network $(G,w)$, set of source nodes $S$, parameters $\varepsilon,k>0$, overlay network $(G''_S,w''_S)$ and source node $s\in S$.
\Ensure Each node $v$ knows $\widetilde d^{4|S|/k}_{G''_S,w''_S}(s,u)$ for each $u\in S$.
\State Let $\mathcal A$ be the Algorithm~\ref{alg:bounded_hop_sssp} with parameters $(G''_S,w''_S,s,4|S|/k,\varepsilon)$ ($\mathcal A$ is of $T=\widetilde O\left(\frac{|S|}{\varepsilon\cdot k}\right)$ rounds, and during the whole computation of $\mathcal A$, each node broadcasts $O(\log n)$ messages to its neighbors due to Lemma~\ref{lem:bounded_hop_sssp}).
\For{$r=1$ to $T$}
\State Let $a$ be the number of nodes in $G''_S$ that want to broadcast a message to its neighbors in $G''_S$ in the $r$-th round of $\mathcal A$.
Count $a$ and make every nodes in $G$ knows $a$ in $O(D_G)$ rounds.
\State Each node in $G''_S$, which wants to send a message to each of its neighbors in $G''_S$, broadcasts such message to all nodes in $G$ (this takes $O(D_G+a)$ rounds).
\EndFor
\end{algorithmic}
\end{algorithm}

\end{document}